\documentclass[10pt,reqno,a4paper]{amsart}
\usepackage[british]{babel}
\usepackage{amssymb,amstext,amsthm,eucal,bbm,mathrsfs,amscd,bbold,pifont}
\usepackage[margin=1in]{geometry}
\usepackage{array}
\usepackage[svgnames,table]{xcolor}
\usepackage[T1]{fontenc}
\usepackage[utf8]{inputenc}
\usepackage{enumitem}
\usepackage[small]{eulervm}
\usepackage{microtype}
\usepackage{verbatim}
\usepackage{booktabs,caption}
\usepackage{appendix}
\usepackage{chngcntr}
\usepackage{apptools}
\usepackage{tikz,pgfplots}
\pgfplotsset{compat=1.14}
\usetikzlibrary{calc,cd}
\usepackage[unicode,pagebackref=true]{hyperref}
\hypersetup{%
  pdftitle   = {Lifshitz symmetry: Lie algebras, spacetimes and particles},
  pdfkeywords = {Lie algebras, kinematical, spacetime, Lifshitz, homogeneous spaces},
  pdfauthor  = {José Figueroa-O'Farrill, Ross Grassie, Stefan Prohazka},
  pdfcreator = {\LaTeX\ with package \flqq hyperref\frqq},
  linkcolor=NavyBlue,
  citecolor=ForestGreen,
  urlcolor=OrangeRed,
  anchorcolor=OrangeRed,
  colorlinks=true
}
\renewcommand*{\backref}[1]{}
\renewcommand*{\backrefalt}[4]{%
  \ifcase #1%
  \or [Page~#2.]%
  \else [Pages~#2.]%
  \fi%
}
\theoremstyle{plain}
\newtheorem{lemma}{Lemma}

\newtheorem{theorem}[lemma]{Theorem}

\theoremstyle{definition}
\newtheorem{definition}{Definition}

\newcommand{\Ggr}{\mathcal{G}}

\newcommand{\eO}{\mathcal{O}}
\newcommand{\g}{\mathfrak{g}}
\newcommand{\h}{\mathfrak{h}}
\renewcommand{\a}{\mathfrak{a}}

\renewcommand{\d}{\partial}

\newcommand{\co}{\mathfrak{co}}

\renewcommand{\r}{\mathfrak{r}}
\newcommand{\n}{\mathfrak{n}}

\newcommand{\so}{\mathfrak{so}}
\newcommand{\iso}{\mathfrak{iso}}
\newcommand{\su}{\mathfrak{su}}
\renewcommand{\k}{\mathfrak{k}}

\newcommand{\s}{\mathfrak{s}}

\renewcommand{\t}{\boldsymbol{t}}

\newcommand{\x}{\boldsymbol{x}}

\ifdefined\C
\renewcommand{\C}{\boldsymbol{C}}
\else
\newcommand{\C}{\boldsymbol{C}}
\fi
\newcommand{\J}{\boldsymbol{J}}

\renewcommand{\P}{\boldsymbol{P}}

\newcommand{\eX}{\mathscr{X}}
\newcommand{\ad}{\operatorname{ad}}
\newcommand{\Ad}{\operatorname{Ad}}

\newcommand{\id}{\mathbb{1}}

\newcommand{\RR}{\mathbb{R}}
\newcommand{\ZZ}{\mathbb{Z}}

\newcommand{\CC}{\mathbb{C}}
\ifdefined\U
\renewcommand{\U}{\operatorname{U}}
\else
\newcommand{\U}{\operatorname{U}}
\fi
\newcommand{\SU}{\operatorname{SU}}

\newcommand{\Aut}{\operatorname{Aut}}
\newcommand{\GL}{\operatorname{GL}}

\newcommand{\CO}{\operatorname{CO}}
\newcommand{\SO}{\operatorname{SO}}
\newcommand{\Spin}{\operatorname{Spin}}
\newcommand{\ISO}{\operatorname{ISO}}

\newcommand{\tr}{\operatorname{tr}}
\newcommand{\EE}{\mathbb{E}}

\renewcommand{\SS}{\mathbb{S}}
\newcommand{\HH}{\mathbb{H}}
\newcommand{\GG}{\mathbb{G}}
\newcommand{\NN}{\mathbb{N}}

\newcommand{\zAdS}{\mathsf{AdS}}

\newcommand{\zS}{\mathsf{S}}
\newcommand{\zTS}{\mathsf{TS}}

\newcommand{\zL}{\mathsf{L}}
\newcommand{\zLW}{\mathsf{LW}}
\newcommand{\zE}{\mathsf{E}}

\definecolor{gris}{rgb}{0.5,0.5,0.5}

\allowdisplaybreaks[0]
\numberwithin{equation}{section}
\AtAppendix {\counterwithin{lemma}{section}}
\begin{document}

\title[Lifshitz symmetry]{Lifshitz symmetry: Lie algebras, spacetimes and particles}
\author[Figueroa-O'Farrill]{José Figueroa-O'Farrill}
\author[Grassie]{Ross Grassie}
\author[Prohazka]{Stefan Prohazka}
\address[JMF,RG,SP]{Maxwell Institute and School of Mathematics, The University
  of Edinburgh, James Clerk Maxwell Building, Peter Guthrie Tait Road,
  Edinburgh EH9 3FD, Scotland, United Kingdom}

\email[JMF]{\href{mailto:j.m.figueroa@ed.ac.uk}{j.m.figueroa@ed.ac.uk}, ORCID: \href{https://orcid.org/0000-0002-9308-9360}{0000-0002-9308-9360}}
\email[RG]{\href{mailto:rgrassie@ed.ac.uk}{rgrassie@ed.ac.uk}, ORCID: \href{https://orcid.org/0000-0002-2073-8841}{0000-0002-2073-8841}}
\email[SP]{\href{mailto:stefan.prohazka@ed.ac.uk}{stefan.prohazka@ed.ac.uk}, ORCID: \href{https://orcid.org/0000-0002-3925-39832e12}{0000-0002-3925-39832e12}}

\begin{abstract}
  We study and classify Lie algebras, homogeneous spacetimes and
  coadjoint orbits (``particles'') of Lie groups generated by spatial
  rotations, temporal and spatial translations and an additional
  scalar generator.  As a first step we classify Lie algebras of this
  type in arbitrary dimension.  Among them is the prototypical
  Lifshitz algebra, which motivates this work and the name ``Lifshitz
  Lie algebras''.  We classify homogeneous spacetimes of Lifshitz Lie
  groups.  Depending on the interpretation of the additional scalar
  generator, these spacetimes fall into three classes:
  \begin{enumerate}
  \item ($d+2$)-dimensional Lifshitz spacetimes which have one
    additional holographic direction;
  \item ($d+1$)-dimensional Lifshitz--Weyl spacetimes which can be
    seen as the boundary geometry of the spacetimes in (1) and where
    the scalar generator is interpreted as an anisotropic dilation;
  \item and $(d+1)$-dimensional aristotelian spacetimes with one
    scalar charge, including exotic fracton-like symmetries that
    generalise multipole algebras.
  \end{enumerate}
  We also classify the possible central extensions of Lifshitz Lie
  algebras and we discuss the homogeneous symplectic manifolds of
  Lifshitz Lie groups in terms of coadjoint orbits.
\end{abstract}

\thanks{EMPG-22-09}

\maketitle
\tableofcontents

\section{Introduction}
\label{sec:introduction}

Since the 1980s, numerous examples of condensed matter systems without
a quasiparticle description have been found~\cite{Hartnoll:2016apf,
  Iqbal:2011ae}. These systems cannot be described using traditional
Landau-Fermi liquid theory; therefore, these non-Fermi liquids, such
as cuprate superconductors~\cite{Arav:2014goa, Taylor:2015glc,
  Mohammadi:2020jbe}, heavy fermion systems near a quantum phase
transition~\cite{Arav:2019tqm, Charmousis:2010zz, Brynjolfsson:2010rx}
and graphene in metallic states~\cite{PhysRevLett.124.076801,
  2018Natur.556...80C,cao2018unconventional} require a new set of
tools for understanding their thermodynamic and transport
properties~\cite{Sachdev:2010ch, Inkof:2019gmh}. In particular, there
is great interest in understanding the universal behaviour of certain
physical properties near quantum critical points~\cite{Inkof:2019gmh,
  MohammadiMozaffar:2012dsy,Erdmenger:2018svl}. For example,
in~\cite{Jeong:2017rxg, PhysRevLett.85.626, Bruin804}, it was shown
that resistivity grows linearly with temperature in heavy fermion
materials. For systems with Lifshitz invariance, Lorentz invariance is
broken, and we observe an anisotropic scaling between time and space,
\begin{equation}
  \label{eq:anisotr}
  t \rightarrow \lambda^z t \quad \text{and} \quad \x \rightarrow \lambda \x,
\end{equation}
where $z \neq 1$. However, away from these critical points, there is
the option that Lorentz invariance is restored, corresponding to $z=1$
above.

Over the last two decades, holography has developed into a novel tool
for investigating the physical properties of condensed matter systems
without a quasiparticle description. A wealth of literature has
appeared under the title AdS/CMT, which seeks to build a holographic
dictionary between the strongly-coupled quantum field theories
underlying these systems and gravitational theories
(see, e.g.,~\cite{Hartnoll:2016apf, Baggio:2011ha, Hartnoll:2009sz,
  Sachdev:2011wg, zaanen_liu_sun_schalm_2015}). This work has led to
exciting connections between black hole physics and the thermal
physics of strongly interacting condensed
matter~\cite{zaanen_liu_sun_schalm_2015}, new ways (see, e.g.,
\cite{Gong:2020pse}) of calculating entanglement entropy in Lifshitz
field theories~\cite{MohammadiMozaffar:2017nri,He:2017wla} and
searching for universal properties of non-relativistic field
theories~\cite{Baggio:2011ha}; it has also led to fields of research
such as holographic superconductors~\cite{Horowitz:2010gk} and
holographic hydrodynamics~\cite{ammon_erdmenger_2015}.

Tailoring this AdS/CMT toolkit to the strongly-coupled systems near
criticality mentioned previously, we arrive at a field of research
known as Lifshitz holography, so-called due to the anisotropic
Lifshitz scaling present in the field theory at the boundary. The
gravitational theories dual to these Lifshitz field theories are built
upon Lifshitz spacetimes, where the scaling symmetry of the boundary
field theory is geometrised as the presence of the radial
dimension~\cite{Kachru:2008yh, Baggio:2011ha}. This transformation of
the scaling symmetry into a radial direction has a nice
interpretation, which was highlighted in~\cite{Hartong:2014pma}.
Namely the Lie algebra of Killing vector fields (KVFs) of the bulk
geometry, which is isomorphic to the Lifshitz algebra, becomes the Lie
algebra of conformal Killing vector fields (CKVFs) of the boundary
geometry. We can view this transformation from KVFs to CKVFs as an
artefact of the boundary geometry being the quotient of the bulk
geometry with respect to the Killing vector field corresponding to the
radial direction.

To see this transformation, consider the $(d+2)$-dimensional Lifshitz
metric~\cite{Kachru:2008yh} (with unit radius of curvature), with
exponent $z$, given in local coordinates
$(r,t,\x) \in (0,\infty) \times \RR \times \RR^d$ by
\begin{equation}
  \label{eq:lifshitz}
  g = - \frac{dt^2}{r^{2z}} + \frac{dr^2}{r^2} +  \frac{d\x^2}{r^2}.
\end{equation}
Here, $d\x^2$ is the euclidean metric on $\RR^d$. For generic values
of the exponent $z$, the Lie algebra of Killing vector fields of the
above metric has dimension $d(d+1)/2 + 2$ and is spanned by
\begin{equation}
  \label{eq:lifshitz-kvfs}
  \xi_{J_{ab}} = -x_a \d_b + x_b \d_a, \qquad \xi_{P_a} = \d_a, \qquad
  \xi_H = \d_t\qquad\text{and}\qquad \xi_D = r \d_r + x^a \d_a + z t \d_t,
\end{equation}
which satisfy the (opposite) Lie brackets to those of the
\emph{Lifshitz Lie algebra} spanned by $J_{ab}= -J_{ba}$, $P_a$, $H$
and $D$:
\begin{equation}\label{eq:lifshitz-lie-algebra}
  \begin{split}
    [J_{ab},J_{cd}] &= \delta_{bc} J_{ad} - \delta_{ac} J_{bd} - \delta_{bd} J_{ac} + \delta_{ad} J_{bc}\\
    [J_{ab}, P_c] &= \delta_{bc} P_a - \delta_{ac} P_b\\
    [D,P_a] &= P_a\\
    [D,H] &= z H,
  \end{split}
\end{equation}
with no other nonzero Lie brackets. It is clear that the Killing
vector fields above span every tangent space, so that the Lie algebra
is transitive and hence the Lifshitz metric is (locally) homogeneous.
In fact, it is not hard to show that the spacetime is homogeneous.
There are two special values of the exponent for which the metric
admits additional Killing vector fields. As mentioned previously, if
$z=1$, then the Lorentz boost invariance of the metric is restored and
\eqref{eq:lifshitz} describes the Poincaré patch of $\zAdS_{d+2}$. In
this case, the space is maximally symmetric, admitting a Lie algebra
of isometries of dimension $(d+2)(d+3)/2$. The other special value is
$z=0$. With this value for the exponent, equation~\eqref{eq:lifshitz}
describes a conformally flat lorentzian metric on the product of a
timelike line with $(d+1)$-dimensional hyperbolic space.

Notice that the Killing vector $\xi_D$ is the only vector field with a
$\d_r$ component. The boundary at $r=0$ is diffeomorphic to the
hypersurface $r = \epsilon$ for some small $\epsilon > 0$, and the
integral curves of $\xi_D$ hit the hypersurface $r=\epsilon$ at
precisely one point. In other words, the hypersurface can be
identified with the space of orbits of $\langle \xi_D \rangle$; that
is, the hypersurface can be identified with the quotient of the
homogeneous Lifshitz spacetime by $\langle \xi_D \rangle$. This
quotient space is what we will call a $(d+1)$-dimensional
Lifshitz--Weyl spacetime in this paper.

This story is analogous, albeit not precisely, to conformal
geometry. One can think of an orientable conformal $n$-dimensional
manifold as one whose frame bundle admits a reduction to the
similitude group $\CO(n) \cong \SO(n) \times \RR^+$.  Here, we
describe an anisotropic version of such a structure: namely, a
($d+1$)-dimensional manifold with a $\CO(d)$-structure. This
modification allows us to break the tangent bundle of such a manifold
as a direct sum of a rank-$d$ vector bundle and a line bundle, and the
dilatation subgroup of $\CO(d)$ assigns a priori different weights to
these bundles.

In this work we provide the first systematic study of Lifshitz Lie
algebras, spacetimes and particles, as well as a first step in the
classification of exotic aristotelian symmetries. We will now provide
a summary of these results.

\section{Summary and overview of the results}
\label{sec:summary-results}

This work can roughly be separated into three parts:
\begin{enumerate}
\item a Lie algebraic part (Section~\ref{sec:lifs-lie-algs}) devoted to
  the classification of Lifshitz Lie algebras;
\item a geometric part (Sections~\ref{sec:homog-lifs-spac}
  and \ref{sec:geom-prop-spac}) devoted to the classification of
  homogeneous spacetimes of Lifshitz Lie groups;
\item and a symplectic part (Section~\ref{sec:coadj-orbits}) which is
  devoted to the partial classification of homogeneous symplectic
  manifolds of Lifshitz Lie groups.
\end{enumerate}
We will now provide a largely self-contained overview and summary of
the main results of this paper and refer to the following sections for
the proofs and further details.

\subsection{Lie algebras}
\label{sec:algebraic-part}

We start with an algebraic classification of \textbf{Lifshitz Lie
  algebras}, which by definition are Lie algebras which contain a
rotational $\so(d)$ subalgebra spanned by $\J$ under which the
remaining generators transform as one vector $\P$ and two scalars $H$
and $D$.  This definition indeed leads to the prototypical Lifshitz
symmetries but, as we will see, also to other interesting algebras and
spacetimes which do not necessarily share the same interpretation. The
definition fixes the Lie brackets
\begin{equation}
  \label{eq:lifshitz-common-intro}
  \begin{split}
    [J_{ab},J_{cd}] &= \delta_{bc} J_{ad} - \delta_{ac} J_{bd} -  \delta_{bd} J_{ac} + \delta_{ad} J_{bc}\\
    [J_{ab}, P_c] &= \delta_{bc} P_a - \delta_{ac} P_b\\
    [J_{ab}, H] &=0\\
    [J_{ab}, D] &=0 \, .
  \end{split}
\end{equation}
The remaining undetermined Lie brackets are only subject to the Jacobi
identity. Up to isomorphism there are seven classes of Lifshitz Lie
algebras, which are listed in Table~\ref{tab:lifshitz} for all
$d \geq 2$. Two of the classes consist of two inequivalent Lie
algebras distinguished by a sign and one of the classes is a
one-parameter ($z$) family, as in equation~\eqref{eq:anisotr}. This
Lie algebra, denoted $\a_{3}^{z}$ here, is the prototypical Lifshitz
algebra and it motivates our generalisations. The details of the
classification are given in
Section~\ref{sec:alas-d>3},~\ref{sec:alas-d=3},~\ref{sec:alas-d=2} and
\ref{sec:summary}. Their interrelations are visualised in
Figure~\ref{fig:limits-aristo-lifshitz} and interpreted geometrically
in Section~\ref{sec:interpr-limits}.

\begin{table}[h!]
\setlength{\tabcolsep}{10pt}
  \centering
  \caption{Lifshitz Lie algebras in $d\geq 2$}
  \label{tab:lifshitz}
  \setlength{\extrarowheight}{2pt}
  \rowcolors{2}{blue!10}{white}
   \resizebox{\linewidth}{!}{
  \begin{tabular}{*{2}{>{$}l<{$}|}*{3}{>{$}l<{$}}|l}
    \toprule
    \multicolumn{1}{c|}{Label} & \multicolumn{1}{c|}{$d$} & \multicolumn{3}{c|}{Nonzero Lie brackets in addition to $[\J,\J] = \J $ and $[\J,\P] = \P$} & \multicolumn{1}{c}{Comments}\\\midrule
    \a_1 & \geq 2 & & & & $\mathfrak{iso}(d) \oplus \RR^2$ \\
    \a_2 & \geq 2 & [D,H] = H & & & \\
    \a_3^z & \geq 2 & [D,H] = z H & [D,P_a]= P_a & & $z \in \RR$\\
    \a_4^\pm & \geq 2 & & & [P_a,P_b]  = \pm J_{ab} &  $\so(d,1) \oplus \RR^2$, $\so(d+1) \oplus \RR^2$ \\
    \a_5^\pm & \geq 2 & [D,H] = H & & [P_a,P_b]=\pm J_{ab} &\\
    \a_6 & 2 & & & [P_a,P_b] = \epsilon_{ab} H & \\
    \a_7 & 2 & [D,H]=2H & [D,P_a]=P_a & [P_a,P_b] = \epsilon_{ab} H & \\
    \bottomrule
  \end{tabular}
}
\end{table}

One might argue that it would be natural to restrict our attention to
Lifshitz Lie algebras where none of the scalar generators, $D$ or $H$,
is central. Under such a restriction we would ignore, as the reader is
free to do, the Lie algebras $\a_{1}$, $\a_{3}^{z=0}$, $\a_{4}^{\pm}$
and $\a_{6}$. We nevertheless keep them for completeness and
uniformity, but we will see that they are mostly (rather trivial)
generalisations of the aristotelian algebras and spacetimes without a
particularly close connection to the prototypical Lifshitz algebra and
spacetime.

A natural algebraic question, especially with our later applications
in mind, is if the Lie algebras admit nontrivial central extensions.
Some do, as we show in Section~\ref{sec:central-extensions} and as
summarised in Table~\ref{tab:lifshitzcent}.

\subsection{Homogeneous spacetimes}
\label{sec:geometric-part}

In Section~\ref{sec:homog-lifs-spac} we classify homogeneous
spacetimes of the Lifshitz Lie algebras. These are smooth
manifolds (with temporal and spatial directions) on which the Lifshitz
Lie groups act transitively.  Intuitively, the spacetimes look
the same at every point, i.e., there are no preferred points.
Well-known homogeneous spacetimes that therefore also have these
properties are the lorentzian Minkowski and (anti)~de~Sitter
spacetimes.

Homogeneous spaces of a Lie group $G$ are described infinitesimally by
Klein pairs $(\g,\h)$ consisting the Lie algebra $\g$ of $G$ and a Lie
subalgebra $\h$ integrating to a closed subgroup of $G$.  Therefore a
practical way to classify homogeneous spaces of $G$ is to classify the
Klein pairs $(\g,\h)$.   In practice this means searching for Lie
admissible subalgebras $\h$ of the Lie algebras in
Table~\ref{tab:lifshitz}.  It is only once we realise the Lie algebra
as vector fields in a homogeneous space that we may assign a
geometric interpretation to the generators in $\g$, be it as
rotations, translations in both space and time or (generalised)
dilatations.  We now discuss the three kinds of homogeneous spacetimes
of Lifshitz Lie groups that we discuss in this paper.

\subsubsection{$(d+2)$-dimensional Lifshitz spacetimes}

Here we quotient by the rotational part of the algebra, i.e., $\h$ is
spanned by $\J$. This leads to $(d+2)$-dimensional spacetimes of
which, loosely speaking and following the discussion in the
Introduction, one direction could be interpreted as being holographic.
These homogeneous Lifshitz spacetimes are summarised in
Table~\ref{tab:d+2-lifshitz-spaces-sum} and there is one for each of
the Lie algebras of Table~\ref{tab:lifshitz}.  The details can be
found in Section~\ref{sec:d+2-lifshitz-spaces}.

\begin{table}[h!]
  \setlength{\tabcolsep}{10pt}
  \centering
  \caption{($d+2$)-dimensional homogeneous Lifshitz spacetimes for $d\geq 2$}
  \label{tab:d+2-lifshitz-spaces-sum}
  \setlength{\extrarowheight}{2pt}
  \rowcolors{2}{blue!10}{white}
  \begin{tabular}{*{3}{>{$}l<{$}|}c}
    \toprule
    \multicolumn{1}{c|}{$\zL$\#}  & \multicolumn{1}{c|}{$\a$} & \multicolumn{1}{c|}{$d$} & \multicolumn{1}{c}{Homogeneous space}\\\midrule
    1 & \a_1 & \geq 2 & $\zS \times \RR \cong \EE^d \times \RR^2$ \\
    2 & \a_2 & \geq 2 & $\EE^d \times \GG$ \\
    3_{z=0} & \a_3^{z=0} & \geq 2 & $\zTS \times \RR$\\
    3_{z\neq0} & \a_3^{z\neq0} & \geq 2 & Lifshitz spacetime\\
    4_\pm & \a_4^\pm & \geq 2 & $\HH^d \times \RR^2$, $\SS^d \times \RR^2$ \\
    5_\pm & \a_5^\pm & \geq 2 & $\HH^d \times \GG$, $\SS^d \times \GG$ \\
    6 & \a_6 & 2 & $\NN \times \RR$\\
    7 & \a_7 & 2 & $\NN$ fibration over $\RR$\\
    \bottomrule
  \end{tabular}
  \vspace{10pt}
  \caption*{$\zS$ is the static aristotelian spacetime, $\EE^d$ is the
    $d$-dimensional euclidean space, $\zTS$ the torsional static
    aristotelian spacetime, $\HH^d$ is $d$-dimensional hyperbolic
    space, $\SS^d$ is the $d$-dimensional sphere, $\GG$ is the
    simply-connected two-dimensional Lie group whose Lie algebra is
    $[D,H] = H$ and $\NN$ is the simply-connected three-dimensional
    Heisenberg group. See Table~\ref{tab:aristotelian-spacetimes} for
    more details on the aristotelian geometries.}
\end{table}
They all share the same invariants of low rank: namely, $H$ and $D$, their
dual one-forms $\eta,\delta$, as well as the degenerate metric $\pi^2$
and degenerate co-metric $P^2$. As we discuss in
Section~\ref{sec:invar-lifsh-spac} one can use these invariants to
construct nondegenerate metrics, in particular we explicitly recover
the prototypical Lifshitz metric~\eqref{eq:lifshitz}.

There is another class of homogeneous spaces that we can interpret as
spacetimes of the Lifshitz algebras.  When we quotient by one
additional scalar generator we obtain $(d+1)$-dimensional spacetimes
and they are the subject of Section~\ref{sec:d+1-lifshitz-spaces}.
Here we must discriminate between two cases depending on whether or not the
Lifshitz Lie group acts effectively.\footnote{An action is effective
  if every non-identity element of the group moves some point.} When
the action is effective, we interpret all symmetries as spacetime
symmetries and we are led to what we term Lifshitz--Weyl spacetimes.

On the other hand, if the action is not effective, then
one of the scalars does not act on the spacetime.  These spacetimes
are aristotelian and we may interpret the additional generator as a
scalar charge, such as, for example, an electric charge.  If the Lie
algebra does not split as a direct sum of an aristotelian algebra and
the one-dimensional subalgebra spanned by the scalar charge, we call
it an exotic spacetime symmetry.

\subsubsection{$(d+1)$-dimensional Lifshitz--Weyl spacetimes}

For Lifshitz--Weyl spacetimes we quotient by one additional generator
($D$ in our notation), to obtain a $(d+1)$-dimensional homogeneous
space. The name of these spacetimes originates from the interpretation
of $D$ as a (generalised) dilatation and we have summarised them in
Table~\ref{tab:d+1-lifshitz-spaces-eff}.

\begin{table}[h!]
  \setlength{\tabcolsep}{10pt} \centering
  \caption{Effective ($d+1$)-dimensional homogeneous Lifshitz--Weyl spacetimes in $d\geq 2$}
  \label{tab:d+1-lifshitz-spaces-eff}
  \setlength{\extrarowheight}{2pt}
  \rowcolors{2}{blue!10}{white}
  \resizebox{\linewidth}{!}{
    \begin{tabular}{*{3}{>{$}l<{$}|}*{3}{>{$}l<{$}}|l}
      \toprule
      \multicolumn{1}{c|}{$\zLW$\#} & \multicolumn{1}{c|}{$\a$} & \multicolumn{1}{c|}{$d$} & \multicolumn{3}{c|}{Nonzero Lie brackets in addition to $[\J,\J] = \J $ and $[\J,\P] = \P$} & \multicolumn{1}{c}{Geometry}\\\midrule
      1 & \a_2 & \geq 2 & [D,H] = H & & & $z=\infty$, flat \\
      2_z & \a_3^z & \geq 2 & [D,H] = z H & [D,\P] = \P & & $z$, flat\\
      3_\pm & \a_5^\pm & \geq 2 &  [D,H] = H & & [\P,\P] = \pm \J & $z=\infty$, not flat\\
      4 & \a_7 & 2 &  [D,H] = 2 H & [D,\P] = \P & [\P,\P] = H & $z=2$, not flat\\
      \bottomrule
    \end{tabular}
  }
  \caption*{The bases are such that the stabiliser subalgebra $\h$ is
    spanned by $\J$ and $D$. These spacetimes admit a canonical
    torsion-free invariant connection which is either flat or not (as
    denoted in the final column).}
\end{table}

These homogeneous spaces can be roughly interpreted as ``going to the
holographic boundary'' with respect to the $(d+2)$ Lifshitz spacetimes
based on the same Lie algebra. As described in
Section~\ref{sec:invar-lifsh-weyl} the invariants of Lifshitz--Weyl
spacetimes are rotational invariant tensors on $M$ which transform
according to some weight. We have summarised these conformal weights
in Table~\ref{tab:d+1-invariants}.

\subsubsection{$(d+1)$-dimensional  aristotelian spacetimes with scalar charge}

For the $(d+1)$-dimensional aristotelian spacetimes we again quotient
by rotations and one scalar charge. In this case this scalar acts
trivially on the spacetime and is therefore not interpretable as a
dilatation or any other spacetime symmetry, but as a scalar
charge $Q$.  Consequently, the underlying geometry is not Lifshitz but
aristotelian.  Aristotelian geometries permit rotational,
spatio-temporal translational symmetries, but no boosts nor dilations
(see~\cite{Figueroa-OFarrill:2018ilb} for a classification
of aristotelian Lie algebras and spacetimes). We have summarised all
of them in Table~\ref{tab:d+1-exot} where one can see that
that they fall into two classes.

\begin{table}[h!]
  \setlength{\tabcolsep}{10pt}
  \centering
  \caption{($d+1$)-dimensional aristotelian spacetimes with one scalar charge in $d\geq 2$}
  \label{tab:d+1-exot}
  \setlength{\extrarowheight}{2pt}
  \rowcolors{2}{blue!10}{white}
  \resizebox{\linewidth}{!}{
    \begin{tabular}{*{3}{>{$}l<{$}|} *{3}{>{$}l<{$}}|l }
      \toprule
      \multicolumn{1}{c|}{$\zE$\#} & \multicolumn{1}{c|}{$\a$} & \multicolumn{1}{c|}{$d$} & \multicolumn{3}{c|}{Nonzero Lie brackets in addition to $[\J,\J] = \J $ and $[\J,\P] = \P$} & \multicolumn{1}{c}{Geometry}                                            \\\midrule
                                   & \a_1                      & \geq 2                   &                                                                                             &             &                  & $\zS$                                  \\
      1                            & \a_2                      & \geq 2                   & [H,Q]= Q                                                                                    &             &                  & $\zS$                                  \\
                                   & \a_3^{z = 0}              & \geq 2                   &                                                                                             & [H,\P] = \P &                  & $\zTS$                                 \\
      2_{z\neq 0}                  & \a_3^{z\neq0}             & \geq 2                   & [H,Q] = zQ                                                                                  & [H,\P] = \P &                  & $\zTS$                                 \\
                                   & \a_4^{\pm}                & \geq 2                   &                                                                                             &             & [\P,\P] = \pm \J & $\HH^d \times \RR$, $\SS^d \times \RR$ \\
      3_{\pm}                      & \a_5^\pm                  & \geq 2                   & [H,Q] = Q                                                                                   &             & [\P,\P] = \pm \J & $\HH^d \times \RR$, $\SS^d \times \RR$ \\
                                   & \a_6                      & 2                        &                                                                                             &             & [\P,\P] = H      & $\NN$                                  \\
      4                            & \a_6                      & 2                        &                                                                                             &             & [\P,\P] = Q      & $\zS$                                  \\
      5                            & \a_7                      & 2                        & [H,Q] = 2Q                                                                                  & [H,\P]= \P  & [\P,\P] = Q      & $\zTS$                                 \\
      \bottomrule
    \end{tabular}
  }
  \caption*{The bases are such that the stabiliser subalgebra $\h$ is
    spanned by $\J$ and $Q$. In all cases the underlying spacetime
    geometry is aristotelian. As denoted in the first column some of
    them are exotic symmetries ($\zE_1$ to $\zE_5$) where the charge
    is influenced by the aristotelian symmetries. The remaining cases
    are such that the Lie algebra is a direct sum of an aristotelian
    algebra and the charge.

    $\zS$ is the static aristotelian spacetime, $\zTS$ the torsional
    static aristotelian spacetime, $\HH^d$ is $d$-dimensional
    hyperbolic space, $\SS^d$ is the $d$-dimensional sphere and $\NN$
    is the simply-connected three-dimensional Heisenberg group. See
    Table~\ref{tab:aristotelian-spacetimes} for more details on the
    aristotelian geometries.}
\end{table}

The first class consists of Lifshitz Lie algebras which are isomorphic
to a direct sum $\langle \J, H, \P \rangle \oplus \RR Q$ of an
aristotelian Lie algebra and the one-dimensional Lie algebra spanned
by the scalar charge $Q$.  The scalar charge decouples from the
spacetime algebra and we think of such symmetries as trivial.

The second class is more interesting.  It consists of Lifshitz Lie
algebra where $Q$ is acted on nontrivially by the spacetime
symmetries. Hence although $Q$ still acts trivially on the spacetime,
its nontrivial interaction with the spacetime symmetries promotes it
to an ``exotic spacetime symmetry'' in an underlying aristotelian
geometry.  We call these spacetime-dependent charges and we say they
are exotic because the ``usual'' (non-exotic) internal charges
decouple from the spacetime symmetries.  In some cases, e.g., $\zE_4$, the
scalar charge is a central extension, but not in all, e.g., $\zE_1$.
It is interesting to note that the only consistent way to get exotic
symmetries in generic dimension is via a charge in the $[H,Q]$
commutator.  We provide further context concerning exotic spacetime
symmetries and their connection to multipole algebras and fractons in
the concluding Section~\ref{sec:conclusion}.

\subsection{Lifshitz particles}
\label{sec:symplectic-part}

Finally, in Section~\ref{sec:coadj-orbits}, we analyse (classical)
\textbf{Lifshitz particles}, equivalently the \textbf{elementary
  systems} with Lifshitz symmetry.  These are homogeneous symplectic
manifolds of a Lifshitz Lie group.  They can be thought as the
mechanistic description of elementary particles.  The geometric
quantisation of these symplectic manifolds give rise (as is well known
for Poincaré and Galilei groups) to unitary irreducible
representations of the group.  As described in more detail in
Section~\ref{sec:coadjoint-orbits} a simply-connected homogeneous symplectic
manifold of a Lie group $G$ is the universal cover of a coadjoint
orbit of $G$ or possibly of a one-dimensional central extension.  One
is therefore led to study the central extensions of the Lie algebras,
see Section~\ref{sec:central-extensions}, and the structure of the
coadjoint orbits.  It is important to emphasise that coadjoint
orbits are an intrinsic property of the Lie group and only once a
choice of homogeneous spacetime is made can they be interpreted as the
space of motions of classical particles in a spacetime.

We construct the coadjoint actions for all cases and highlight the
structure of the coadjoint orbits. We put special emphasis on the
prototypical case $\a_{3}^{z\neq 0}$ and provide the necessary
information to perform the classification, if so desired.

\section{Lifshitz Lie algebras}
\label{sec:lifs-lie-algs}

In this section we classify Lifshitz Lie algebras, but first a
definition.\footnote{In \cite{Figueroa-OFarrill:2018ygf}, one of us
  introduced the notion of a ``generalised Lifshitz algebra'' to be a
  graded kinematical Lie algebra together with the grading element.
  The Lifshitz Lie algebras in the present paper will be seen to be
  graded aristotelian Lie algebras together with the grading element.
  This means that the ``generalised Lifshitz algebras'' in
  \cite{Figueroa-OFarrill:2018ygf} are a special class of
  boost-extended Lifshitz Lie algebras, which will be the subject of a
  follow-up paper.}

\begin{definition}
  A \emph{Lifshitz Lie algebra} is a ($d(d+1)/2 +
  2$)-dimensional real Lie algebra $\g$ with basis $J_{ab}=-J_{ba},
  P_a, H, D$, for $a,b =1,\dots,d$, such that $J_{ab}$ span a Lie
  subalgebra $\r \cong \so(d)$, under which $P_a$ transforms as a
  vector and $H$ and $D$ transform as scalars.
\end{definition}

Unpacking this definition, we see that all such Lie algebras share the
following Lie brackets:
\begin{equation}
  \label{eq:lifshitz-common}
  \begin{split}
    [J_{ab},J_{cd}] &= \delta_{bc} J_{ad} - \delta_{ac} J_{bd} -  \delta_{bd} J_{ac} + \delta_{ad} J_{bc}\\
    [J_{ab}, P_c] &= \delta_{bc} P_a - \delta_{ac} P_b\\
    [J_{ab}, H] &=0\\
    [J_{ab}, D] &=0,
  \end{split}
\end{equation}
and the additional Lie brackets not involving $J_{ab}$ are subject
only to the Jacobi identity.

If $d=1$ then there are no rotations and hence any three-dimensional
Lie algebra is Lifshitz. These were classified by
Bianchi~\cite{Bianchi} and we will not discuss them further in this
paper, except briefly when discussing the associated spacetimes. It
will be convenient to break the discussion of the $d>1$ algebras into
three cases: $d = 2$, $d=3$ and $d > 3$. For the following it will be
useful to recall the aristotelian spacetimes.

The classification of aristotelian Lie algebras and aristotelian
homogeneous spacetimes was done in \cite{Figueroa-OFarrill:2018ilb},
which we recall in Table~\ref{tab:aristotelian-spacetimes}. The
notation is as in \cite[Table 2]{Figueroa-OFarrill:2018ilb}, except
that we have introduced $\NN$ (as in ``Nilmanifold'') for the Heisenberg
group which in \cite{Figueroa-OFarrill:2018ilb} was called
``A24''. Since every aristotelian Lie algebra gives rise to a unique
aristotelian Lie pair, this Table provides both the Lie algebras and
the classification of simply-connected aristotelian spacetimes up to
isomorphism.

\begin{table}[h!]
  \centering
  \caption{Aristotelian Lie algebras and simply-connected homogeneous
    ($d+1$)-dimensional aristotelian spacetimes}
  \label{tab:aristotelian-spacetimes}
  \rowcolors{2}{blue!10}{white}
  \begin{tabular}{l|>{$}c<{$}|*{2}{>{$}l<{$}}|l}
    \toprule
    \multicolumn{1}{c|}{Name} & d & \multicolumn{2}{c|}{Nonzero Lie
                                    brackets in addition to $[\J,\J] =
                                    \J $ and $[\J,\P] = \P$}& \multicolumn{1}{c}{Comments}\\
    \midrule
    $\zS$ & \geq 0 & & & static \\
    $\zTS$ & \geq 1 & [H,\P] = \P & & torsional static\\
    $\RR\times\HH^d$ & \geq 2 & & [\P,\P] = \J & \\
    $\RR\times\SS^d$ & \geq 2 & & [\P,\P] = - \J & \\
    $\NN$ & 2 & & [\P,\P] = H & Heisenberg group\\
    \bottomrule
  \end{tabular}
  \caption*{The bases are such that the subalgebra $\h$ of the
    spacetimes is spanned by $\J$.}
\end{table}

\subsection{Lifshitz algebras with $d>3$}
\label{sec:alas-d>3}

If $d>3$, then the $\so(d)$-equivariance of the brackets ---
equivalently, the Jacobi identity involving $J_{ab}$ --- implies that
the brackets involving $D$ are given by
\begin{equation}
  [D,H] = a H + b D \qquad\text{and}\qquad [D,P_a] = \lambda P_a.
\end{equation}
In particular, $D$ and $H$ span a two-dimensional Lie subalgebra.  Up
to isomorphisms, there are precisely two such Lie algebras and we can
therefore assume that $[D,H] = a H$, where $a \in \{0,1\}$.
In any case we see that $[D,-]$ is diagonal and hence $D$ is a grading
element.  Since $D$ does not appear in the RHS of any Lie brackets,
we conclude that $D$ is a grading element for the aristotelian Lie
subalgebra spanned by the remaining generators: $J_{ab},P_a, H$.

\subsection{Lifshitz algebras with $d=3$}
\label{sec:alas-d=3}

Now let $d=3$.  The $\so(3)$-equivariance of the Lie brackets
restricts their form to
\begin{equation}
  \label{eq:lifshitz-d=3}
  \begin{split}
    [H,P_a] &= \alpha \epsilon_{abc} J^{bc} + \beta P_a\\
    [Z,P_a] &= \gamma \epsilon_{abc} J^{bc} + \delta P_a\\
    [D,H] &= a H\\
    [P_a, P_b] &= \varepsilon \epsilon_{abc} P^c + \eta J_{ab}.
  \end{split}
\end{equation}
The terms with coefficients $\alpha,\gamma,\varepsilon$ are unique to
$d=3$.  We will show that they can be set to zero, which extends the
results of $d>3$ also to $d=3$.  We start with the observation that by
shifting $P_a \mapsto P_a + \tfrac\varepsilon4 \epsilon_{abc} J^{bc}$,
we can (and will) set $\varepsilon=0$ without loss of generality.
Since we are in $d=3$, it is more convenient to change basis to
$J_{ab} = -\epsilon_{abc} J^c$, where $J_a = -\tfrac12 \epsilon_{abc}
J^{bc}$.  In this way, $[J_a,J_b] = \epsilon_{abc} J^c$, et cetera.
Independently of the $[D,H]$ bracket, the $[H,\P,\P]$ Jacobi forces
$\alpha = 0$ and the $[D,\P,\P]$ Jacobi forces $\gamma =0$.  Since
$\alpha = \gamma = 0$, we see that the possible new terms in $d=3$ do
not arise and hence the results for $d>3$ also apply to $d=3$.  In
other words, $J_a, P_a, H$ span an aristotelian Lie algebra and $D$ is
a grading element.

\subsection{Lifshitz algebras with $d=2$}
\label{sec:alas-d=2}

We now consider the case $d=2$.  Now $J_{ab} = - \epsilon_{ab} J$ spans
$\r \cong \so(2)$, which is abelian.  The (non-zero) brackets are now
\begin{equation}
  \label{eq:lif-lie-alg-d=2}
  \begin{split}
    [J,P_a] &= \epsilon_{ab} P_b\\
    [P_a,P_b] &= \epsilon_{ab} (\alpha J + \beta H + \gamma D)\\
    [D,H] &= a J + b H + c D \\
    [D,P_a] &= \delta P_a + \xi \epsilon_{ab} P_b\\
    [H,P_a] &= \theta P_a + \psi \epsilon_{ab} P_b.
  \end{split}
\end{equation}
The first observation is that we can set $\xi=\psi=0$ without loss of
generality by redefining $D \mapsto D - \xi J$ and $H \mapsto H - \psi
J$.  Having done that, there are three Jacobi identities which need to
be checked: $[H,\P,\P]$, $[D,\P,\P]$ and $[D,H,\P]$.  (The $[\P,\P,\P]$ Jacobi identity is
identically satisfied because there are only two $P_a$.)  Since $[D,-]$ and
$[H,-]$ act diagonally on $P_a$, $[[D,H],-]$ must act trivially.  This
says that $J$ cannot appear in $[D,H]$ and hence $a = 0$.  In
addition, we see that $b\theta + c \delta = 0$.  Since $a = 0$, $D,H$
span a two-dimensional Lie algebra and hence we can change basis so
that $[D,H] = b H$ where $b \in \{0,1\}$.  Therefore again we see that
$D$ is a grading element for the aristotelian Lie algebra spanned by
$J,P_a,H$.

\subsection{Summary}
\label{sec:summary}

In summary, we have seen that every Lifshitz Lie algebra
is obtained by adjoining a grading element to an aristotelian Lie
algebra.  It is a now simple matter to determine the possible
gradings of the aristotelian Lie algebras in
Table~\ref{tab:aristotelian-spacetimes}.

\begin{theorem}\label{thm:lif-lie-alg}
  Every Lifshitz Lie algebra is isomorphic to precisely one of the Lie
  algebras in Table~\ref{tab:lifshitz}.
\end{theorem}

\begin{proof}
  As explained above, we need to determine the possible gradings of
  the aristotelian Lie algebras.  These Lie algebras are in one-to-one
  correspondence with the homogeneous aristotelian spacetimes in
  Table~\ref{tab:aristotelian-spacetimes}.  We go through them in
  turn, using as labels the names of the corresponding spacetimes.

  \begin{itemize}
  \item[($\zS$)]  Any grading is possible: $[D,\P]=\lambda \P$ and
    $[D,H]= \mu H$.  We distinguish three cases:
    \begin{enumerate}
    \item  If $\lambda = \mu = 0$, we obtain $\a_1$ in
      Table~\ref{tab:lifshitz}, which is simply the direct sum $\s
      \oplus \RR D$, where $\s$ is the static aristotelian Lie
      algebra.
    \item If $\lambda =0$ but $\mu \neq 0$, we may rescale $D \mapsto
      \mu^{-1}D$ and we are left with $[D,\P]=0$ and $[D,H]=H$, which
      could be thought of as the case $z=\infty$ of the Lifshitz Lie
      algebra~\eqref{eq:lifshitz-lie-algebra}.  We label it $\a_2$ in
      Table~\ref{tab:lifshitz}.
    \item If $\lambda \neq 0$, we may rescale $D \mapsto \lambda^{-1}
      D$ so that $[D,\P]=\P$ and, letting $z := \mu/\lambda$, $[D,H]=z
      H$.  This is the Lifshitz algebra in
      equation~\eqref{eq:lifshitz-lie-algebra}, which we label
      $\a_3^z$ in Table~\ref{tab:lifshitz}.
    \end{enumerate}
  \item[($\zTS$)]  In this case $[D,H]=0$ and $[D,\P]=\lambda\P$.
    We must distinguish two cases, according to whether or not
    $\lambda=0$.
    \begin{enumerate}
    \item If $\lambda=0$, we may rename $H \leftrightarrow D$
      to arrive at a Lie algebra isomorphic to $\a_3^{z=0}$.
    \item If $\lambda \neq 0$, we may rescale $D \mapsto
      \lambda^{-1}D$ and arrive at $[D,\P]=\P$ and then change basis
      $H \mapsto H-D$, so that $[H,\P]=0$. The resulting Lie
      algebra is also isomorphic to $\a_3^{z=0}$.
    \end{enumerate}
  \item[($\RR \times \HH^d$)]  Here $[D,\P]=0$, so we have $[D,H] =
    \mu H$ and we must distinguish two cases depending on whether or not $\mu
    =0$.  As in the previous case, $\mu  =0$ says $D$ is central and
    gives a direct sum Lie algebra, whereas if $\mu\neq 0$, we may
    rescale $D$ so that $[D,H]=H$.  The former is labelled $a_4^+$ in
    Table~\ref{tab:lifshitz} and the latter is labelled $a_5^+$.
  \item[($\RR \times \SS^d$)]  This case is similar to the above and
    leads to the Lie algebras $a_4^-$ and $a_5^-$ in
    Table~\ref{tab:lifshitz}.
  \item[($\NN$)] The only compatible grading is $[D,\P]=\lambda P$ and
    $[D,H]=2 \lambda H$.  If $\lambda = 0$ we have the Lie algebra
    $\a_6$ in Table~\ref{tab:lifshitz} and if $\lambda \neq 0$, we may
    rescale $D \mapsto \lambda^{-1}D$ so that $[D,\P]=\P$ and
    $[D,H]=2H$, which is Lie algebra $\a_7$.
  \end{itemize}
\end{proof}

The Lifshitz Lie algebras are related via contractions which are
depicted in Figure~\ref{fig:limits-aristo-lifshitz}. Since the
homogeneous Lifshitz spacetimes are in one-to-one correspondence with
the Lifshitz Lie algebras, that figure also depicts limits between the
spacetimes.  Every node in the figure represents an isomorphism class
of Lifshitz Lie algebras, with the red nodes only arising when $d=2$.
The thin edges correspond to contractions between the algebras,
whereas the thick edge is the continuum $\a_3^z$.  We see that
$\lim_{z\to\pm\infty} \a_3^z = \a_2$, so that $z$ should really be
thought of parametrising a circle.  We can make this manifest by
defining $z = \cot\tfrac\theta 2$ with $\theta \in [0,2\pi]$.  There
are two distinguished points in this circle: $\theta = 2 \cot^{-1}2$,
corresponding to $z=2$, and $\theta = 0$ (equivalently, $\theta =
2\pi$) which corresponds to $z  = \infty$.

\begin{figure}[h!]
  \centering
  \begin{tikzpicture}[>=latex, shorten >=2.5pt, shorten <=2.5pt, x=1cm,y=1cm]
    %
    %
    %
    %
    \coordinate [label=above:{\small $\a_1$}] (a1) at (0,2);
    \coordinate [label=below:{\small $\a_2$}] (a2) at (0,0);
    \coordinate [label=above left:{\small $\a_4^\pm$}] (a4) at (-1,2); 
    \coordinate [label=above left:{\small $\a_5^\pm$}] (a5) at (-2,1); 
    \coordinate [label=below left:{\small $\a_3^2$}] (a3) at (0.8,-0.4);
    \coordinate [label=above right:{\color{red}\small $\a_6$}] (a6) at (1,2);
    \coordinate [label=above right:{\color{red}\small $\a_7$}] (a7) at (2,1);
    %
    %
    \draw[-,line width=3pt,color=green!70!black] (0,-1) circle (1);
    %
    %
    \draw [->,line width=0.5pt,color=blue] (a5) -- (a4); 
    \draw [->,line width=0.5pt,color=blue] (a5) -- (a2); 
    \draw [->,line width=0.5pt,color=blue] (a4) -- (a1); 
    \draw [->,line width=0.5pt,color=blue] (a2) -- (a1); 
    \draw [->,line width=0.5pt,color=red] (a7) -- (a6); 
    \draw [->,line width=0.5pt,color=red] (a6) -- (a1); 
    \draw [->,line width=0.5pt,color=red] (a7) -- (a3); 
    %
    %
    \draw (0,-2) node[below] {\color{green!30!black}\small $\a_3^z$};
    %
    %
    \foreach \point in {a1,a2,a3,a4,a5}
    \filldraw [color=blue!70!black,fill=blue!50!white] (\point) circle (2pt);
    \foreach \point in {a6,a7}
    \filldraw [color=red!70!black,fill=red!50!white] (\point) circle (2pt);
  \end{tikzpicture}  
  \caption{Contractions between Lifshitz Lie algebras}
  \caption*{(The red nodes and their corresponding arrows only exist if $d=2$.)}
  \label{fig:limits-aristo-lifshitz}
\end{figure}

\subsection{Central extensions}
\label{sec:central-extensions}

As we will recall in Section~\ref{sec:coadj-orbits}, homogeneous
symplectic manifolds (elementary particles in the language of Souriau
\cite{MR1461545}) of a Lie group $G$ are given locally by coadjoint
orbits of $G$ or possibly of a one-dimensional central extension of
$G$. As a first step in the determination of homogeneous symplectic
manifolds of the Lifshitz Lie groups, we will work out the central
extensions of the Lifshitz Lie algebras in Table~\ref{tab:lifshitz}.

Central extensions of a Lie algebra $\g$ are classified by the second
Chevalley--Eilenberg cohomology group $H^2(\g)$ (with values in the
trivial one-dimensional representation).  Since a Lifshitz Lie algebra
$\g$ contains a rotational subalgebra $\r$ which acts reducibly on the
Chevalley--Eilenberg complex and trivially on the cohomology,
$H^2(\g)$ can be calculated from the (typically much smaller)
subcomplex consisting of $\r$-invariant cochains.

For the purposes of this paper, namely the determination of
homogeneous symplectic manifolds of a Lifshitz Lie group there is one
additional subtlety. Not every central extension of the Lie algebra
integrates to a central extension of the Lie group $G$; although they
always do if we were to take $G$ to be simply connected. For example,
it follows from $[J,D]=Z$, that the adjoint action of $J$ on $D$
integrates to
\begin{equation}
  \Ad_{\exp(\theta J)} D = \exp(\theta \ad_J) D = D + \theta Z,
\end{equation}
which is not periodic in $\theta$. Hence it is not a central extension
of the Lie group $G$ where the rotational group is $\SO(2)$; although
it is a central extension of the universal covering group where the
``rotation'' subgroup is $\RR$. As explained briefly in
Section~\ref{sec:coadj-orbits}, we restrict\footnote{Let us however
  note that we could envision interesting physics for the case where
  one drops this restriction: for example, Lifshitz anyons in $d=2$ as
  in \cite{Jokela:2016nsv}.} to compact rotations and therefore these
central extension are not relevant for the determination of
homogeneous symplectic manifolds with symmetry group $G$.  This means
we can actually restrict ourselves not just to the $\r$-invariant
subcomplex, but even further to the $\r$-basic subcomplex which
calculates the relative cohomology $H^2(\g;\r)$. We choose to
calculate $H^2(\g)$ below, since this may be of independent interest,
but later in the paper we shall only be interested in some of these
central extensions.

By introducing a suitable parameter, we may treat the Lie algebras
$\a_1$, $\a_2$ and $\a_3^z$ together.  Similarly we may treat $\a_4^\pm$
and $\a_5^\pm$ together and also $\a_6$ and $\a_7$ together.  The
results are summarised in Table~\ref{tab:lifshitzcent}, where central
extensions in $H^2(\g)$ which are not in $H^2(\g;\r)$ have been
parenthesised and we shall ignore them in the remainder of the paper.

\begin{table}[h!]
\setlength{\tabcolsep}{10pt}
  \centering
  \caption{Nontrivial central extensions of the Lifshitz Lie algebras in Table~\ref{tab:lifshitz}}
  \label{tab:lifshitzcent}
  \resizebox{\linewidth}{!}{
  \setlength{\extrarowheight}{2pt}
  \begin{tabular}{>{$}l<{$} | >{$}r<{$} | *{2}{>{$}c<{$}|} *{4}{>{$}l<{$}}}
    \toprule
    \multicolumn{1}{c|}{$\g$}       & \multicolumn{1}{c|}{$d$} &  \multicolumn{1}{c|}{$\dim H^2(\g)$} & \multicolumn{1}{c|}{$\dim H^2(\g,\r)$} & \multicolumn{4}{c}{Nonzero Lie brackets of central extensions}       \\\midrule
    \rowcolor{blue!10} \a_1          & \geq 3                   & 1 & 1                                                        & [D,H] =Z &                               &             &             \\
    \rowcolor{blue!10}               & 2                        & 4 & 2                                                        & [D,H]=Z  & [P_a,P_b] = \epsilon_{ab} Z_P & ([J,D]=Z_D) & ([J,H]=Z_H) \\
    \a_2                             & \geq 3                   & 0 & 0                                                        &          &                               &                           \\
                                     & 2                        & 2 & 1                                                        &          & [P_a,P_b] = \epsilon_{ab} Z_P & ([J,D]=Z_D) &             \\
\rowcolor{blue!10}    \a_3^{z\neq 0} & \geq 3                   & 0 & 0                                                        &          &                               &             &             \\
\rowcolor{blue!10}                   & 2                        & 1 & 0                                                        &          &                               & ([J,D]=Z_D) &             \\
     \a_3^{z=0}                      & \geq 3                   & 1 & 1                                                        & [D,H] =Z &                               &             &             \\
                                     & 2                        & 3 & 1                                                        & [D,H]=Z  &                               & ([J,D]=Z_D) & ([J,H]=Z_H) \\
\rowcolor{blue!10}    \a_4^\pm       & \geq 2                   & 1 & 1                                                       & [D,H] =Z &                               &             &             \\
     \a_5^\pm                        & \geq 2                   & 0 & 0                                                        &          &                               &             &             \\
\rowcolor{blue!10}    \a_6           & 2                        & 1 & 0                                                       &          &                               & ([J,D]=Z_D) &             \\
     \a_7                            & 2                        & 1 & 0                                                        &          &                               & ([J,D]=Z_D) &             \\
    \bottomrule
  \end{tabular}
}
\caption*{This table shows the central extensions of the Lifshitz
  algebras $\g$. The dimension of the cohomology group $H^2(\g)$
  depends in some cases on the dimension $d$ of the Lie algebra $\g$.
  For compact rotations, the central extensions in parentheses do not
  integrate to the group and are related to the relative cohomology
  $H^2(\g;\r)$, where $\r$ is spanned by the rotations.}
\end{table}

\subsubsection{Central extensions of $\a_1$, $\a_2$ and $\a_3^z$}
\label{sec:centr-extens-a_1}

We introduce a parameter $w \in \{0,1\}$ and modify the Lie brackets by $[D,P] = w
P$.  We can then treat all three Lie algebras simultaneously with
$\a_1$ corresponding to $w=z=0$, $\a_2$ corresponding to $w=0,z=1$ and
$\a_3^z$ corresponding to $w=1$.  Any other choice of $w,z$ is related
to one of these by rescaling the generator $D$.

Let $\g$ (depending on $w,z$) be one of these Lie algebras.  A basis
for $\g$ is $J_{ab}, P_a, H, D$ and the canonical dual basis for
$\g^*$ is $\lambda^{ab},\pi^a,\eta,\delta$.  We will use the notation
$\left<\cdots\right>$ to mean the real span of the enclosed vectors,
so that $\g = \left<J_{ab}, P_a, H, D\right>$ and
$\g^*=\left<\lambda^{ab},\pi^a,\eta,\delta\right>$.  We let $\r =
\left<J_{ab}\right>$.  We are interested in the
following fragment of the $\r$-invariant subcomplex of the
Chevalley--Eilenberg complex
\begin{equation}
  \begin{tikzcd}
    (\g^*)^\r \arrow[r,"\d"] & (\wedge^2\g^*)^\r \arrow[r,"\d"] & (\wedge^3\g^*)^\r,
  \end{tikzcd}
\end{equation}
where
\begin{equation}\label{eq:r-inv-1-cochains}
  (\g^*)^\r = \left< \eta, \delta, \tfrac12 \epsilon_{ab}\lambda^{ab}\right>
\end{equation}
and
\begin{equation}\label{eq:r-inv-2-cochains}
  (\wedge^2\g^*)^\r = \left<\eta \wedge \delta,
    \tfrac12\epsilon_{ab}\lambda^{ab} \wedge \delta,
    \tfrac12\epsilon_{ab}\lambda^{ab} \wedge \eta,
    \tfrac12\epsilon_{ab}\pi^a \wedge \pi^b, \tfrac12 \epsilon_{abc}
    \lambda^{ab} \wedge \pi^c, \tfrac18 \epsilon_{abcd}\lambda^{ab}
    \wedge \lambda^{cd}\right>,
\end{equation}
and where a term involving the Levi-Civita $\epsilon$ symbol only
appears in the relevant dimension.

The Chevalley--Eilenberg differential on generators is given by
\begin{equation}
  \d \lambda^{ab} = - \lambda^a{}_c \wedge \lambda^{cb}, \qquad \d
  \pi^a = - \lambda^a{}_b \wedge \pi^b - w \delta \wedge \pi^a, \qquad
  \delta \eta = -z \delta \wedge \eta \qquad\text{and}\qquad \d \delta
  = 0.
\end{equation}
We see that the space of 2-coboundaries $B^2(\g)^\r = \d (\g^*)^\r \subset
(\wedge^2\g^*)^\r$ is given by
\begin{equation}
  B^2(\g)^\r =
  \begin{cases}
    \left<\delta \wedge \eta \right> & \text{if $z \neq 0$,}\\
    0 & \text{if $z = 0$.}
  \end{cases}
\end{equation}

We calculate $\d: (\wedge^2\g^*)^\r \to (\wedge^3\g^*)^\r$ to obtain
\begin{equation}
  \begin{split}
    \d (\delta \wedge \eta) &= 0\\
    \d (\tfrac12 \epsilon_{ab}\lambda^{ab} \wedge \delta) &= 0\\
    \d (\tfrac12 \epsilon_{ab} \lambda^{ab} \wedge \eta) &= z  \epsilon_{ab} \lambda^{ab} \wedge \delta \wedge \eta\\
    \d (\tfrac12 \epsilon_{ab} \pi^a \wedge \pi^b) &= -2 w \epsilon_{ab} \pi^a \wedge \pi^b \wedge \delta\\
    \d (\tfrac12 \epsilon_{abc}\lambda^{ab} \wedge \pi^c) &= -2 \epsilon_{abc} \lambda^{ab} \wedge \pi^c \wedge \delta\\
    \d (\tfrac18 \epsilon_{abcd} \lambda^{ab} \wedge \lambda^{cd}) &= -2 \epsilon_{abcd} \lambda^a{}_e \wedge \lambda^{eb} \wedge \lambda^{cd}.
  \end{split}
\end{equation}
We see that the space of 2-cocycle $Z^2(\g)^\r = \ker \d :
(\wedge^2\g^*)^\r \to (\wedge^3\g^*)^\r$ is given by
\begin{equation}
  Z^2(\g)^\r =
  \begin{cases}
    \left<\delta \wedge \eta, \tfrac12 \epsilon_{ab} \pi^a \wedge \pi^b, \tfrac12 \epsilon_{ab}\lambda^{ab} \wedge \delta,  \tfrac12 \epsilon_{ab} \lambda^{ab} \wedge \eta\right> & \text{if $w=z=0$}\\
    \left<\delta \wedge \eta, \tfrac12 \epsilon_{ab} \pi^a \wedge  \pi^b, \tfrac12 \epsilon_{ab}\lambda^{ab} \wedge \delta\right> & \text{if $w=0$ and $z \neq 0$}\\
    \left<\delta \wedge \eta, \tfrac12 \epsilon_{ab} \lambda^{ab} \wedge \eta, \tfrac12 \epsilon_{ab} \lambda^{ab} \wedge \delta\right> & \text{if $w=1$ and $z=0$}\\
    \left<\delta \wedge \eta, \tfrac12 \epsilon_{ab}\lambda^{ab}  \wedge \delta\right> & \text{if $w=1$ and $z \neq 0$.}
  \end{cases}
\end{equation}

Therefore we see that $H^2(\g) = Z^2(\g)^\r / B^2(\g)^\r$ behaves
quite differently in $d=2$ and $d \geq 3$.  The latter is given by
\begin{equation}
  H^2(\g)_{d\geq 3} =
  \begin{cases}
    \left<[\delta \wedge \eta]\right> & \text{if $z=0$}\\
    0 & \text{if $z\neq 0$,}
  \end{cases}
\end{equation}
where here and in the sequel square brackets denotes the cohomology
class of the enclosed cocycle.  We see that for $d\geq 3$,
$\a_2$ and $\a_3^{z\neq 0}$ do not admit any nontrivial central
extensions, whereas $\a_1$ and $\a_3^{z=0}$ admit a one-dimensional
nontrivial central extension with bracket $[D,H] = Z$.

If $d=2$ things are more complicated:
\begin{equation}
  H^2(\g)_{d=2} =
  \begin{cases}
    \left<[\tfrac12 \epsilon_{ab} \lambda^{ab} \wedge \delta]\right> & \text{if $w=1$ and $z\neq 0$}\\
    \left<[\tfrac12 \epsilon_{ab} \lambda^{ab} \wedge \delta],  [\tfrac12\epsilon_{ab}\pi^a \wedge \pi^b]\right> & \text{if $w=0$ and $z \neq 0$}\\
    \left<[\delta \wedge \eta], [\tfrac12 \epsilon_{ab} \lambda^{ab} \wedge \delta], [\tfrac12 \epsilon_{ab} \lambda^{ab} \wedge \eta]\right> & \text{if $w=1$ and $z=0$}\\
    \left<[\delta \wedge \eta], [\tfrac12 \epsilon_{ab} \lambda^{ab}
      \wedge \delta], [\tfrac12 \epsilon_{ab} \lambda^{ab} \wedge
      \eta],  [\tfrac12\epsilon_{ab}\pi^a \wedge \pi^b]\right> &
    \text{if $w=z=0$.}
  \end{cases}
\end{equation}
So that if $d=2$, the Lifshitz Lie algebra $\a_3^{z\neq 0}$ admits a
nontrivial central extension with additional bracket $[J,D]=Z_D$; the
Lie algebra $\a_2$ admits a two-dimensional space of nontrivial
central extensions with brackets $[J,D]=Z_D$ and $[P_a,P_b]=
\epsilon_{ab} Z_P$; the Lie algebra $\a_3^{z=0}$ admits a
three-dimensional space of nontrivial central extensions with brackets
$[D,H]=Z$, $[J,D]=Z_D$ and $[J,H]=Z_H$; and $\a_1$ admits a
four-dimensional space of nontrivial central extensions with brackets
$[D,H]=Z$, $[J,D]=Z_D$, $[J,H]=Z_H$ and $[P_a,P_b]=\epsilon_{ab} Z_P$.

\subsubsection{Central extensions of $\a_4^\pm$ and $\a_5^\pm$}
\label{sec:centr-extens-a_4}

We treat these two Lie algebras simultaneously by introducing a
parameter $z \in \{0,1\}$ and defining the bracket $[D,H]=zH$.  If
$z=0$ we are in $\a_4^\pm$ and if $z=1$ we are in $\a_5^\pm$.  The
bases for $\g$ and $\g^*$ are as in the previous section and the
spaces of $\r$-invariant $1$- and $2$-cochains are as before and given
in equations~\eqref{eq:r-inv-1-cochains} and
\eqref{eq:r-inv-2-cochains}, respectively.  The Chevalley--Eilenberg
differential is of course different and given on generators by
\begin{equation}
  \d \lambda^{ab} = - \lambda^a{}_c \wedge \lambda^{cb} \mp \pi^a
  \wedge \pi^b, \qquad \d \pi^a = - \lambda^a{}_b \wedge \pi^b, \qquad
  \d \eta = - z \delta \wedge \eta \qquad\text{and}\qquad \d \delta = 0.
\end{equation}
So that the space of 2-coboundaries is now given by
\begin{equation}
  B^2(\g)^\r = \left<z \delta \wedge \eta, \tfrac12 \epsilon_{ab}
    \pi^a \wedge \pi^b\right>.
\end{equation}
We calculate the differential on the 2-cochains and obtain
\begin{equation}
  \begin{split}
    \d (\delta \wedge \eta) &= 0 \\
    \d (\tfrac12 \epsilon_{ab} \lambda^{ab} \wedge \delta) &= \mp \tfrac12 \epsilon_{ab} \pi^a \wedge \pi^b \wedge \delta\\
    \d (\tfrac12 \epsilon_{ab} \lambda^{ab} \wedge \eta) &= \mp \tfrac12 \epsilon_{ab} \pi^a \wedge \pi^b \wedge \eta + \tfrac12 z \epsilon_{ab} \lambda^{ab} \wedge \delta \wedge \eta\\
    \d (\tfrac12 \epsilon_{ab} \pi^a \wedge \pi^b) &= 0\\
    \d (\tfrac12 \epsilon_{abc} \lambda^{ab} \wedge \pi^c &= \mp  \tfrac12 \epsilon_{abc} \pi^a \wedge \pi^b \wedge \pi^c\\
    \d (\tfrac18 \epsilon_{abcd} \lambda^{ab} \wedge \lambda^{cd}) &=  -\tfrac14 \epsilon_{abcd} \lambda^a{}_e \wedge \wedge \lambda^{eb}
    \wedge \lambda^{cd} \mp \tfrac14 \epsilon_{abcd} \pi^a \wedge \pi^b \wedge \lambda^{cd},
  \end{split}
\end{equation}
so that the space of $2$-cocycles is
\begin{equation}
  Z^2(\g)^\r = \left<\eta \wedge \delta, \tfrac12 \epsilon_{ab} \pi^a  \wedge \pi^b\right>.
\end{equation}
In summary, the cohomology is then
\begin{equation}
  H^2(\g) =
  \begin{cases}
    \left<[\eta \wedge \delta]\right> & \text{if $z=0$}\\
    0 & \text{if $z\neq 0$.}
  \end{cases}
\end{equation}
Hence $\a_4^\pm$ admits a one-dimensional non-trivial central
extension with bracket $[D,H]=Z$ and $\a_5^\pm$ admits none.

\subsubsection{Central extensions of $\a_6$ and $\a_7$}
\label{sec:centr-extens-a_6}

Finally, we introduce a parameter $w \in \{0,1\}$ in the brackets
$[D,P]=wP$ and $[D,H]=2w H$ so that if $w=0$ we are in $\a_6$ and if
$w=1$ we are in $\a_7$.  We are in $d=2$ here, so that the
$\r$-invariant $1$- and $2$-cochains are
\begin{equation}
  (\g^*)^\r = \left<\eta,\delta, \tfrac12 \epsilon_{ab} \lambda^{ab}\right>
\end{equation}
and
\begin{equation}
  (\wedge^2\g^*)^\r = \left<\eta \wedge \delta, \tfrac12 \epsilon_{ab}
    \lambda^{ab} \wedge \delta, \tfrac12 \epsilon_{ab} \lambda^{ab}
    \wedge \eta, \tfrac12 \epsilon_{ab} \pi^a \wedge \pi^b\right>.
\end{equation}
The action of the differential on generators is now
\begin{equation}
  \d \lambda^{ab} = 0, \qquad \d \eta = -2 w \delta \wedge \eta -
  \tfrac12 \epsilon_{ab} \pi^a \wedge \pi^b, \qquad \d \pi^a = -
  \lambda^a{}_b \wedge \pi^b - w \delta \wedge \pi^a
  \qquad\text{and}\qquad \d \delta = 0.
\end{equation}
The 2-coboundaries are then
\begin{equation}
  B^2(\g)^\r = \left<\tfrac12 \epsilon_{ab} \pi^a \wedge \pi^b + 2 w \delta \wedge \eta\right>.
\end{equation}

We calculate the differential on $\r$-invariant cochains to be
\begin{equation}
  \begin{split}
  \d (\eta \wedge \delta) &= - \tfrac12 \epsilon_{ab} \pi^a \wedge \pi^b \wedge \delta\\
  \d (\tfrac12 \epsilon_{ab} \lambda^{ab} \wedge \eta)&= w \epsilon_{ab} \lambda^{ab} \wedge \delta \wedge \eta + \tfrac12 \lambda^{ab} \wedge \pi_a \wedge \pi_b\\
  \d(\tfrac12 \epsilon_{ab} \lambda^{ab} \wedge \delta) &= 0\\
  \d (\tfrac12 \epsilon_{ab} \pi^a \wedge \pi^b) &= - w \epsilon_{ab} \pi^a \wedge \pi^b \wedge \delta.
\end{split}
\end{equation}
We see that the space of 2-cocycles is given by
\begin{equation}
  Z^2(\g)^\r = \left<\tfrac12 \epsilon_{ab} \pi^a \wedge \pi^b + 2 w  \delta \wedge \eta, \tfrac12 \epsilon_{ab} \lambda^{ab} \wedge \delta\right>,
\end{equation}
so that the cohomology is given by
\begin{equation}
  H^2(\g) = \left<[\tfrac12 \epsilon_{ab} \lambda^{ab} \wedge \delta]\right>.
\end{equation}
Therefore both $\a_6$ and $\a_7$ admit a one-dimensional nontrivial
central extension with bracket $[J,D]= Z_D$.

\section{Spatially isotropic homogeneous Lifshitz spacetimes}
\label{sec:homog-lifs-spac}

In this section we classify the homogeneous spacetimes associated to
the Lifshitz algebras in Table~\ref{tab:lifshitz}.  As discussed in
the introduction, we are interested in two kinds of homogeneous
spacetimes. Firstly, we have the $(d+2)$-dimensional Lifshitz
spacetimes which are described infinitesimally by Klein pairs of the
form $(\a,\r)$, where $\a$ is a Lifshitz algebra and $\r \cong \so(d)$
is the rotational subalgebra, which are the subject of
Section~\ref{sec:d+2-lifshitz-spaces}.  Secondly, we have the
$(d+1)$-dimensional Lifshitz--Weyl spacetimes whose Klein pairs are
now $(\a,\h)$, where $\a$ is again a Lifshitz algebra, but now
$\h \cong \co(d) = \so(d) \oplus \RR$ is spanned by the rotations and
one one of the scalars in the span of $D,H$.  This is the subject of
Section~\ref{sec:d+1-lifshitz-spaces}.

\subsection{Homogeneous Lifshitz spacetimes}
\label{sec:d+2-lifshitz-spaces}

We now list the possible Klein pairs $(\a,\r)$ where $\a$ is one of
the Lie algebras of Table~\ref{tab:lifshitz} and $\r \cong \so(d)$
is the subalgebra spanned by $J_{ab}$.  Clearly these are indexed by
the Lie algebras $\a$ in Table~\ref{tab:lifshitz} themselves.  We
may easily identify the corresponding homogeneous spaces, partially
from the classification in
\cite[Appendix~A]{Figueroa-OFarrill:2018ilb}.  The results are
tabulated in Table~\ref{tab:d+2-lifshitz-spaces-sum}.  The notation is such
that $\HH^d$ is $d$-dimensional hyperbolic space, $\SS^d$ is the
$d$-dimensional sphere, $\EE^d$ is the $d$-dimensional euclidean
space, $\zS$ is the static aristotelian spacetime, $\zTS$ the
torsional static aristotelian spacetime, $\NN$ is the simply-connected
three-dimensional Heisenberg group (a three-dimensional aristotelian
spacetime prosaically labelled $A24$ in
\cite{Figueroa-OFarrill:2018ilb}) and $\GG$ is the simply-connected
two-dimensional Lie group whose Lie algebra is $[D,H] = H$.  For
$d=1$, since $\r=0$, the Klein pairs are simply the Lie algebras
themselves, which are the Bianchi Lie algebras.  The spacetimes are
the simply-connected three-dimensional Lie groups, as in Bianchi's
original paper \cite{Bianchi}.

The simply-connected four-dimensional homogeneous Lifshitz spacetime
with Klein pair $(\a_7,\r)$ may be identified with the
simply-connected Lie group $\Ggr$ whose Lie algebra $\g$ is an
extension-by-derivation of the three-dimensional Heisenberg Lie
algebra $\n$, so fitting into an exact sequence
\begin{equation}
  \begin{tikzcd}
    0 \arrow[r] & \n \arrow[r] & \g \arrow[r] & \RR D \arrow[r] & 0.
  \end{tikzcd}
\end{equation}
The group $\Ggr$ is foliated by copies of the Heisenberg group $\NN$
and fibres over the real line.

\subsection{Homogeneous Lifshitz--Weyl and aristotelian spacetimes with scalar charge}
\label{sec:d+1-lifshitz-spaces}

We now list the possible Klein pairs $(\a,\h)$ where $\a$ is one of
the Lie algebras of Table~\ref{tab:lifshitz} and
$\h = \r \oplus \RR S$ is the subalgebra spanned by $J_{ab}$ and a
scalar $S$ in the span of $D$ and $H$. These Klein pairs describe
homogeneous Lifshitz--Weyl spacetimes or aristotelian spacetimes with
one scalar charge. For each Lie algebra $\a$ in
Table~\ref{tab:lifshitz}, we determine the possible one-dimensional
scalar lines (spanned by $S$) up to the action of $\J$-preserving
automorphisms of $\a$; that is, automorphisms of $\a$ which are the
identity on $\r$. For $d=1$, the Klein pairs were already classified
(from the point of view of kinematical spacetimes) in
\cite[Section~3.4]{Figueroa-OFarrill:2018ilb} and further studied in
\cite{Figueroa-OFarrill:2019sex}.

\subsubsection{Klein pairs associated to $\a_1$}
\label{sec:klein-pairs-a1}

The $\J$-preserving automorphisms are given by
\begin{equation}
  \label{eq:autos-a1}
  \P \mapsto \mu \P, \qquad H \mapsto a H + b D \qquad\text{and}\qquad
  D \mapsto c H + d D,
\end{equation}
where $\mu \neq 0$ and $\begin{pmatrix} a & b \\ c & d\end{pmatrix}$
is invertible. Clearly, we can take any $S$ to $D$ via an
automorphism.  The resulting Klein pair is not effective, since $D$
spans an ideal both of $\a$ and of $\h$.  Quotienting by this ideal,
we obtain a Klein pair $(\s,\r)$ where $\s$ is the aristotelian static
Lie algebra and $\r$ is the rotational subalgebra.  As shown in
\cite{Figueroa-OFarrill:2018ilb}, this is the Klein pair of the static
aristotelian spacetime $\zS$.

\subsubsection{Klein pairs associated to $\a_2$}
\label{sec:klein-pairs-a2}

The $\J$-preserving automorphisms are given by
\begin{equation}
  \label{eq:autos-a2}
  \P \mapsto \mu \P, \qquad H \mapsto a H \qquad\text{and}\qquad
  D \mapsto D + c H,
\end{equation}
where $\mu, a \neq 0$ and $c \in \RR$. Suppose that $S = \alpha H +
\beta D$.  Then under such an automorphism, $S \mapsto (a \alpha +
\beta c) H + \beta D$.  If $\beta \neq 0$, then we can choose $c$ so
that $S = \beta D$ and if $\beta = 0$, then $S = \alpha H$.  In other
words, we can take $\h = \r \oplus \RR D$ or $\h = \r \oplus \RR H$.
The latter Klein pair is not effective, since $H$ spans an ideal of
both $\a$ and $\h$.  Quotienting by this ideal gives the Klein pair of
the aristotelian static spacetime $\zS$.  The former Klein
pair is effective and describes a homogeneous flat Lifshitz spacetime
with $z=\infty$.

\subsubsection{Klein pairs associated to $\a_3^z$}
\label{sec:klein-pairs-a3z}

Here the $\J$-preserving automorphisms are given by
\begin{equation}
  \label{eq:autos-a3z}
  \P \mapsto \mu \P, \qquad H \mapsto a H \qquad\text{and}\qquad
  D \mapsto D + c H,
\end{equation}
where $\mu, a \neq 0$ and $c \in \RR$, and as before we have two
choices of scalar lines up to automorphisms: $\h = \r \oplus \RR D$ or
$\h = \r \oplus \RR H$.  The latter Klein pair is not effective, since
$H$ spans an ideal of both $\a$ and $\h$.  Quotienting by this ideal
gives the Klein pair of the aristotelian torsional static spacetime
$\zTS$.  The former Klein pair is effective and describes a
homogeneous flat Lifshitz--Weyl spacetime with scaling exponent $z\neq
0$.

\subsubsection{Klein pairs associated to $\a_4^\pm$}
\label{sec:klein-pairs-a4pm}

The $\J$-preserving automorphisms are given by
\begin{equation}
  \label{eq:autos-a4pm}
  \P \mapsto \pm \P, \qquad H \mapsto a H + b D \qquad\text{and}\qquad
  D \mapsto c H + d D,
\end{equation}
where $\begin{pmatrix} a & b \\ c & d\end{pmatrix}$ is invertible.
We may take $S = D$ without loss of generality, resulting in a
non-effective Klein pair describing $\HH^d \times \RR$ or $\SS^d
\times \RR$.

\subsubsection{Klein pairs associated to $\a_5^\pm$}
\label{sec:klein-pairs-a5}

The $\J$-preserving automorphisms are given by
\begin{equation}
  \label{eq:autos-a5}
  \P \mapsto \pm \P, \qquad H \mapsto a H \qquad\text{and}\qquad
  D \mapsto D + c H,
\end{equation}
where $a \neq 0$ and $c \in \RR$.  There are two possibilities, namely
$S = H$ and $S=D$.  If $S=H$, the resulting Klein pair is
non-effective and quotienting by the ideal generated by $H$ gives
either $\HH^d \times \RR$ or $\SS^d \times \RR$.  If we take $S = D$
then we get an effective Klein pair corresponding to a generalised
Lifshitz spacetime with $z = \infty$, but this time with curvature.

\subsubsection{Klein pairs associated to $\a_6$}
\label{sec:klein-pairs-a6}

The $\J$-preserving automorphisms are given by
\begin{equation}
  \label{eq:autos-a6}
  \P \mapsto \lambda \P, \qquad H \mapsto \lambda^2 H
  \qquad\text{and}\qquad D \mapsto a D + b H,
\end{equation}
with $a,\lambda$ nonzero.  As before, there are two possibilities:
$S=D$ and $S=H$.  Neither case is effective, quotienting by the ideal
generated by $D$ we obtain the Heisenberg group $\NN$ as an
aristotelian spacetime, whereas quotienting by the ideal generated by
$H$ we obtain the static aristotelian spacetime $\zS$.

\subsubsection{Klein pairs associated to $\a_7$}
\label{sec:klein-pairs-a7}

The $\J$-preserving automorphisms are given by
\begin{equation}
  \label{eq:autos-a7}
  \P \mapsto \lambda \P, \qquad H \mapsto c H \qquad\text{and}\qquad D \mapsto a D + b H,
\end{equation}
where $a,c,\lambda$ are nonzero.  As in the previous case, there are
two possibilities: $S = H$ and $S=D$.  If we take $S=H$, the resulting
Klein pair is not non-effective and quotienting by the ideal generated
by $H$ recovers the torsional static aristotelian spacetime $\zTS$.
Taking $S=D$ we get an effective Klein pair describing a
Lifshitz--Weyl geometry with $z = 2$ and nonzero curvature.

\subsubsection{Summary}
\label{sec:summary-spaces}

We summarise this discussion in
Tables~\ref{tab:d+1-lifshitz-spaces-eff} and \ref{tab:d+1-exot}, which
list the homogeneous Lifshitz--Weyl and aristotelian spacetimes with
scalar charge, respectively. We list the Klein pairs $(\a,\h)$ and we
have changed basis so that the scalar $S$ in $\h$ is denoted $D$ in
the case it acts effectively or $Q$ in case it does not.

The Lifshitz--Weyl spacetimes are all both reductive and symmetric.
Being symmetric homogeneous spaces, they have a canonical torsion-free
invariant connection, which is flat in the first two cases ($\zLW1$
and $\zLW2_z$) and not flat in the next two ($\zLW3_\pm$ and $\zLW4$).

\subsection{Geometric interpretation of the limits}
\label{sec:interpr-limits}

Having understood the nature of the homogeneous Lifshitz spacetimes,
we may now give a geometric interpretation of some of the Lie algebra
contractions in Figure~\ref{fig:limits-aristo-lifshitz}.
The contractions $\a_5^\pm \to \a_2$ are flat limits in that the
round metric on $\SS^d$ and the hyperbolic metric on $\HH^d$ flatten
to become the euclidean metric on $\EE^d$.  This interpretation also
holds for the $(d+1)$-dimensional homogeneous aristotelian spacetimes
with scalar charge.  For the $(d+1)$-dimensional Lifshitz--Weyl
spacetimes, the interpretation is slightly different.  In this case
the flat limit refers to the flatness of the canonical torsion-free
invariant connection.

The contractions $\a_5^\pm \to \a_4^\pm$ and $\a_2 \to \a_1$ are such
that the nonabelian Lie group $\GG$ becomes abelian.  They may be
understood geometrically as an aristotelian limit of the Lifshitz
geometries: essentially in this limit the action of the scale
transformations becomes trivial.  A similar interpretation can be
given to the contraction $\a_7 \to \a_6$, where the four-dimensional
Lie group $\Ggr$, which is a semidirect product $\RR^+ \ltimes \NN$
becomes a direct product $\RR^+ \times \NN$ again via the
trivialisation of the scale transformations.

\section{Geometrical properties of the spacetimes}
\label{sec:geom-prop-spac}

\subsection{Invariants of Lifshitz spacetimes}
\label{sec:invar-lifsh-spac}

The Lifshitz spacetimes in Table~\ref{tab:d+2-lifshitz-spaces-sum} are all
homogeneous spaces of the form $G/H$ where $H\cong \SO(d)$.  Weyl
\cite[Theorem~2.11.A]{MR1488158} proved that the primitive tensor invariants
of $\SO(d)$, out of which any other invariant tensor can be written,
are the Kronecker $\delta_{ab}$ and the Levi-Civita
$\epsilon_{ab\cdots c}$.  Therefore all homogeneous Lifshitz
spacetimes in Table~\ref{tab:d+2-lifshitz-spaces-sum} share the same
invariant tensors.  In low rank, they are given by vector fields $H$
and $D$, their dual one-forms $\eta,\delta$, as well as $\pi^2 =
\delta_{ab}\pi^a \pi^b$ and $P^2 = \delta^{ab} P_a P_b$.  In addition
we have the corresponding volume forms.  Notice that each of these
spaces admits invariant metrics of signatures $(d+2,0)$, $(d+1,1)$ and
$(d,2)$; although perhaps it is the lorentzian case which is the most
relevant in the present context.  Even in this case, there is of
course a choice: e.g., $\pi^2 + 2 \eta \delta$ and $\pi^2 + \delta^2 -
\eta^2$ give rise to different invariant lorentzian metrics.

We will now show that the Lifshitz metric~\eqref{eq:lifshitz} is
indeed one of the invariant metrics of the homogeneous Lifshitz
spacetime~$\zL3_{z}$. Parametrising the group element as
$\sigma(t,\x,\rho)=e^{t H + \x \cdot \P} e^{\rho D}$ we can calculate the
pull back of the (left-invariant) Maurer-Cartan form $\vartheta$ (for
the details of this computation see, e.g., \cite[Section
3.6.]{Figueroa-OFarrill:2019sex})
\begin{equation}
\sigma^{*} \vartheta = e^{-z \rho} dt H + e^{-\rho} d \x \cdot \P +   d \rho D = \theta \, .
\end{equation}
Since the canonical invariant connection vanishes it is also
equivalent to the soldering form $\theta$. We can now use the
soldering form to map the invariant tensors to the tangent space of
the manifold, e.g., $\theta(\eta)=e^{-z \rho} dt$,
$\theta(\delta)=d \rho$, and $\theta(\pi^{a}) = e^{-\rho} d x^{a}$. We
can now write $- \eta^2 + \delta^2 + \pi^2 $ in coordinates
\begin{equation}
  - e^{-2z \rho} dt^{2} + d\rho^{2} + e^{-2 \rho} d \x \cdot d \x \, , 
\end{equation}
which with the change of coordinates $r=e^{\rho}$ leads us to the
Lifshitz metric \eqref{eq:lifshitz}.

\subsection{Invariants of Lifshitz--Weyl spacetimes}
\label{sec:invar-lifsh-weyl}

For the $(d+1)$-dimensional Lifshitz--Weyl spacetimes in
Table~\ref{tab:d+1-lifshitz-spaces-eff} -- that is, those
$(d+1)$-dimensional spaces corresponding to effective Klein pairs --
the natural invariants are not tensors but what we could term
conformal classes of tensors. These Lifshitz--Weyl spacetimes are
quotients of the Lifshitz spacetimes in
Table~\ref{tab:d+2-lifshitz-spaces-sum} by the one-parameter subgroup
generated (in the notation of Table~\ref{tab:d+1-lifshitz-spaces-eff}) by
$D$. As homogeneous spacetimes, the Lifshitz--Weyl spacetimes are
diffeomorphic to $G/H$, where
$H \cong \CO(d)\cong \SO(d) \times \RR^+$ is the $d$-dimensional
similitude group. In particular, they are base manifolds for a
principal $H$-bundle for which all the tensor bundles are associated
bundles. For example, the tangent bundle of a Lifshitz--Weyl spacetime
$M$ is associated to the reducible representation of $\CO(d)$ given by
$(V \otimes L^\lambda) \oplus (S \otimes L^\mu)$, where $V$ and $S$
are the vector and scalar representations of $\SO(d)$, respectively,
and $L^w$ is the one-dimensional representation of the subgroup of
dilatations of weight $w$. Here, $\lambda$ and $\mu$ are the
$D$-weights of $P_a$ and $H$, respectively. The natural invariant
tensors on $M$ are then rotational invariants which transform
according to some weight. The result of Weyl \cite{MR1488158} quoted
above says that the rotational invariant tensors of low rank are the
vector $H$, the dual one-form $\eta$, the symmetric rank-2 covariant
tensor $\pi^2$ and the symmetric bivector $P^2$. In
Table~\ref{tab:d+1-invariants} we tabulate the ``conformal'' weights
of these low-rank invariant tensors for the Lifshitz--Weyl spacetimes.
It follows from this table that the only such spacetime with a
conformal structure (either lorentzian $\pi^2 - \eta^2$ or riemannian
$\pi^2 + \eta^2$) is $\zLW2_{z=1}$, corresponding to conformally
compactified Minkowski spacetime, since $\zL2_{z=1}$ is the Poincaré
patch of $\zAdS$ and $\zLW_{z=1}$ is the quotient by the dilatations,
which as mentioned in the introduction is diffeomorphic to the
conformal boundary.

\begin{table}[h!]
\setlength{\tabcolsep}{10pt}
  \centering
  \caption{Conformal weights of low-rank invariants of homogeneous Lifshitz--Weyl spacetimes}
  \label{tab:d+1-invariants}
  \setlength{\extrarowheight}{2pt}
  \rowcolors{2}{blue!10}{white}
  \begin{tabular}{>{$}l<{$}|*{4}{>{$}r<{$}}}
    \toprule
    & \multicolumn{4}{c}{Weights} \\
    \multicolumn{1}{c|}{$\zLW$\#}& \multicolumn{1}{c}{$H$}  & \multicolumn{1}{c}{$\eta$} & \multicolumn{1}{c}{$P^2$} & \multicolumn{1}{c}{$\pi^2$} \\
    \midrule
    1 & 1 & -1 & 0 & 0 \\
    2_z & z & -z & 2 & -2 \\
    3_\pm & 1 & -1 & 0 & 0 \\
    4 & 2 & -2 & 2 & -2\\
    \bottomrule
  \end{tabular}
\end{table}

\subsection{Invariants of the aristotelian spacetimes}
\label{sec:invar-arist-spac}

Since the action of the charge $Q$ on the geometry is not effective
the invariants are the same as for the aristotelian spacetimes without
the scalar charge. In particular they admit all the aforementioned
invariants $H$, $\eta$, $P^{2}$ and $\pi^{2}$ as true invariants, not
just as conformal ones.

\section{Lifshitz particles}
\label{sec:coadj-orbits}

We now shift attention to another class of homogeneous manifolds of
Lifshitz Lie groups, the Lie groups of the Lifshitz Lie algebras in
Table~\ref{tab:lifshitz}.  They are not to be interpreted as
spacetimes, but as elementary systems (loosely, particles) with
Lifshitz symmetry; that is, symplectic manifolds admitting a
transitive action of a Lifshitz group via symplectomorphisms.  Let us
briefly review the relationship between homogeneous symplectic
manifolds and coadjoint orbits of certain central extensions.  A more
detailed description motivated by the present paper can be found in
\cite{BFOSymplectic}.

\subsection{Coadjoint orbits}
\label{sec:coadjoint-orbits}

Let $G$ be a connected Lie group acting transitively on a
simply-connected symplectic manifold $(M,\omega)$ via
symplectomorphisms.  If we let $g \in G$ also denote the
diffeomorphism of $M$ induced by $g$, then this condition is simply
$g^*\omega = \omega$.  As shown by Souriau \cite{MR1461545},
associated to such data there is a moment map $\mu : M \to \g^*$, with
$\g$ the Lie algebra of $G$, defined up to the addition of a constant
element of $\g^*$ and such that it is $G$-equivariant: intertwining between
the $G$-action on $M$ and an affinisation of the coadjoint action of
$G$ on $\g^*$; that is, for all $g \in G$ and $p \in M$,
\begin{equation}
  \mu(g \cdot p) = \Ad_g^* \mu(p) + \theta(g),
\end{equation}
where $\theta : G \to \g^*$ is a symplectic group cocycle.  In other
words, $\theta$ obeys the cocycle condition
\begin{equation}
  \theta(g_1g_2) = \theta(g_1) + \Ad_{g_1}^* \theta(g_2)
\end{equation}
and its derivative $d_e\theta : \g \to \g^*$ at the identity is such
that
\begin{equation}
  \left<(d_e\theta)(X),Y\right> = - \left<(d_e\theta)(Y),X\right>.
\end{equation}
If $\theta$ were a coboundary, so that $\theta(g) = \mu_0 - \Ad_g^*
\mu_0$ for some constant $\mu_0 \in \g^*$, then we could redefine the
moment map: $\mu \mapsto \mu' = \mu-\mu_0$, so that now
\begin{equation}
  \mu'(g \cdot p ) = \Ad_g^* \mu'(p)
\end{equation}
is equivariant relative to the (linear) coadjoint action.  If $\theta$
is cohomologically nontrivial, then it defines a one-dimensional
central extension $\widehat G$ of $G$
(see \cite[Theorem~15]{BFOSymplectic}) and the affine action of $G$ on
$\g^*$ is now essentially the coadjoint representation of $\widehat G$
on $\widehat\g^*$.  It then follows that $M$ is the universal cover of
a coadjoint orbit of $\widehat G$.

One-dimensional central extensions\footnote{Strictly speaking with
  kernel $\cong \RR$, but the topology of the kernel is of no
  consequence.} of $G$ are classified by the smooth group cohomology
group $H^2(G)$.  The celebrated van~Est theorem
\cite{MR0059285,MR0070959} (see also \cite{MR4081118}) implies that
$H^2(G)$ is isomorphic to the relative Lie algebra cohomology
$H^2(\g,\k)$, where $\k$ is the Lie algebra of a maximal compact
subgroup of $G$.  Hence to determine (up to coverings) the homogeneous
symplectic manifolds of $G$ we need to determine the coadjoint orbits
of every one-dimensional central extension of $G$ whose van~Est
derivative defines a class in $H^2(\g,\k)$.

As an illustration of how coadjoint orbits arise from particle motion,
let us consider briefly geodesic motion in the Lifshitz spacetime
$(M^{d+2},g)$ where $g$ is the metric tensor in
equation~\eqref{eq:lifshitz}.  For generic values of $z$ (here $z \neq
0,1$), $g$ admits Killing vector fields given by
equation~\eqref{eq:lifshitz-kvfs} generating an action of the Lifshitz
group $G$ with Lie algebra $\g = \a_3^z$ on $M$ by isometries.  Let $\gamma : I
\to M$ be an affinely-parametrised geodesic of the Levi-Civita
connection defined by $g$.  Every such geodesic $\gamma$ defines an
element $\alpha_\gamma \in \g^*$; that is, a linear map $\alpha_\gamma : \g
\to \RR$ given by $X \mapsto g(\dot\gamma,\xi_X)$, which is a constant
along the geodesic.  If we let $a \in G$, let $\phi_a : M \to M$
denote the corresponding diffeomorphism.  Then $\phi_a \circ \gamma : I \to M$ is also an
affinely parametrised geodesic and it is not hard to show that
$\alpha_{\phi_a \circ \gamma} = \Ad_a^* \alpha_\gamma$.  Therefore a
geodesic $\gamma$ defines a map $G \to \g^*$, sending $a \in G$ to
$\alpha_{\phi_a \circ \gamma} = \Ad^*_a \alpha_\gamma$, which is none other
but the coadjoint orbit map.  The image of this map is precisely the
coadjoint orbit of $\alpha_\gamma$.

Let us observe that this assignment from particle trajectories in a
spacetime to coadjoint orbits is not bijective, in that different
spacetimes might give rise to the same coadjoint orbit.  This simply
reflects the fact that coadjoint orbits are an intrinsic property of
the Lie group and not a property of the spacetime.

In this section we present some partial results on the calculation of
coadjoint orbits of some Lifshitz Lie groups.  We divide the
discussion into parts labelled by the Lie algebras in
Table~\ref{tab:lifshitz}.

\subsection{$\a_1$, $\a_2$ and $\a_3^z$}
\label{sec:a_1-a_2-a_3z}

As in Section~\ref{sec:centr-extens-a_1}, we treat them together by
introducing a new parameter $w$ and declaring $[D,P_a] = w
P_a$. Rescaling $D$, we see that $(w,z)$ are defined only up to
multiplication by a nonzero real number. The choice $(z,w)=(0,0)$
gives $\a_1$, whereas $\a_2$ corresponds to $(z,w)=(1,0)$ and $\a_3^z$
corresponds to $(z,w)=(z,1)$.

\subsubsection{Adjoint and coadjoint actions}
\label{sec:adjo-coadj-acti}

Let $\g$ be the Lie algebra spanned by $J_{ab},P_a,H,D$ subject to the
brackets \eqref{eq:lifshitz-common} together with
\begin{equation}
  \label{eq:g-brackets}
    [P_a,P_b] = 0, \qquad [H, P_a] = 0, \qquad [D,P_a] = w P_a \qquad\text{and}\qquad  [D,H] = z H.
\end{equation}

When discussing coadjoint orbits associated to a Lie algebra $\g$ we
need to specify the Lie group $G$ under consideration.  One way is
to take $G$ to be the unique connected and simply-connected group with
Lie algebra $\g$.  This group typically does not act effectively on
$\g$ (and hence on its dual), but a certain quotient (known as the
adjoint group) does.  For example, for $\g = \su(n)$, one would take $G
= \SU(n)$ and the adjoint group is the quotient by the $\ZZ_n$
subgroup consisting of scalar matrices.  For the Lie algebra $\g$
under consideration, defined by the
brackets~\eqref{eq:lifshitz-common} and \eqref{eq:g-brackets}, we may
take $G$ to be the group with underlying manifold $\RR^+ \times \SO(d)
\times \RR^d \times \RR$ and multiplication given by
\begin{equation}
  \label{eq:G-mult}
  (\sigma_1, A_1, v_1, h_1) \cdot (\sigma_2, A_2, v_2, h_2) =
  (\sigma_1 \sigma_2, A_1 A_2, v_1 + \sigma_1^w A_1 v_2, h_1 +
  \sigma_1^z h_2).
\end{equation}
The group $G$ is not simply-connected: its universal cover would have
underlying manifold $\RR^+ \times \Spin(d) \times \RR^d \times \RR$
(if $d>2$) or $\RR^+ \times \RR^2 \times \RR^2 \times \RR$ (if $d=2$), but all the
representations of $\so(d)$ appearing in the Lie algebra 
are tensorial and hence factor through $\SO(d)$, so that the adjoint
group is $G$, at least when $z\neq 0$.  If $z=0$ then the adjoint
representation has nontrivial kernel and the adjoint group is the
quotient of $G$ by this kernel.

It follows from the multiplication law \eqref{eq:G-mult} that the
identity of $G$ is $(1,\id,0,0)$ and inversion is given by
\begin{equation}
  \label{eq:g-inverse}
  (\sigma, A, v, h)^{-1} = (\sigma^{-1}, A^{-1}, -\sigma^{-w} A^{-1} v, -\sigma^{-z} h).
\end{equation}
As a check, let us calculate the Lie algebra of $G$.  Consider a curve
$\gamma(t) = (\sigma(t), A(t), v(t), h(t))$ in $G$ with $\gamma(0) =
(1,\id,0,0)$ and $\gamma'(0) = (\lambda, X, p, \varepsilon) \in \RR \oplus
\so(d) \oplus \RR^d \oplus \RR$.

The adjoint action on $g \in G$ on $\gamma'(0) \in \g$ is given by the
velocity at the identity of the curve $g \gamma(t) g^{-1}$.  If we let
$g = (\sigma, A, v, h)$ and $\gamma'(0)=(\lambda, X,q,\theta)$ as
before, we find that
\begin{equation}
  \label{eq:adjoint-action}
  \Ad_{(\sigma, A, v, h)} (\lambda, X,q,\theta) = (\lambda, A X
  A^{-1}, \sigma^w A q - w \lambda v - A X A^{-1}v, \sigma^z
  \theta - z \lambda h).
\end{equation}
It follows from this expression that
\begin{equation}
  \label{eq:ker-ad}
  \ker\Ad =
  \begin{cases}
    \{(1,\id,0,0)\}, &  \text{if $z \neq 0$}\\
   \left\{ (1,\id,0,h) \middle | h \in \RR\right\}, & \text{if $z = 0$ and $w \neq 0$}\\
    \left\{ (\sigma,\id,0,h) \middle | \sigma \in \RR^+,\ h \in \RR\right\}, & \text{if $z=w=0$.}
  \end{cases}
\end{equation}

Replacing $(\sigma, A, v, h )$ by a curve
$(\sigma(s), A(s), v(s), h(s))$ in $G$ through the identity in
equation~\eqref{eq:adjoint-action}, we obtain a curve in $\g$ through
$(\lambda,X,q,\theta)$, whose velocity there is the bracket
\begin{equation}
  [(\sigma'(0), A'(0), v'(0), h'(0)), (\lambda, X, q, \theta)].
\end{equation}
Performing this calculation we obtain (after changing notation)
\begin{equation}
  \label{eq:g-Lie-algebra}
  [(\lambda_1, X_1, q_1, \theta_1), (\lambda_2, X_2, q_2,
  \theta_2)] = \left(0, [X_1,X_2], w(\lambda_1 q_2 -\lambda_2 q_1) +
  X_1 q_2 - X_2 q_1, z(\lambda_1 \theta_2 - \lambda_2 \theta_1)\right).
\end{equation}
Comparing with \eqref{eq:lifshitz-common} and \eqref{eq:g-brackets},
we see that
\begin{equation}
  \label{eq:g-LA-map}
  J_{ab} = (0, E_{ab}, 0,0), \qquad P_a =
  (0,0,e_a,0), \qquad H = (0,0,0,1) \qquad\text{and}\qquad D = (1, 0, 0, 0),
\end{equation}
where $e_a$ are the elementary vectors in $\RR^d$ and $E_{ab} \in \so(d)$ is the
skew-symmetric endomorphism defined by
\begin{equation}
  \label{eq:skew-endos}
  E_{ab} e_c = \delta_{bc} e_a - \delta_{ac} e_b.
\end{equation}

We may identify $\g^*$ with $\g$ as vector spaces, under the inner
product $\left<-,-\right> : \g \times \g \to \RR$ defined by
\begin{equation}
  \label{eq:g-dual-pair}
  \left< (\mu, Y, p, \varepsilon), (\lambda, X, q, \theta)\right> =
  \lambda \mu + \tfrac12 \tr(Y^T X) + p^T q + \theta\varepsilon.
\end{equation}
In this way, we can identify the canonical dual basis
$\lambda^{ab},\pi^a, \eta, \delta$ for $\g^*$ with
\begin{equation}
  \label{eq:g-dual-basis}
  \lambda^{ab} = (0, E_{ab}, 0,0), \qquad \pi^a = (0,0,e_a, 0), \qquad
  \eta = (0,0,0,1)\qquad\text{and}\qquad \delta = (1,0,0,0).
\end{equation}
This inner product is not invariant under the adjoint representation,
so that the adjoint and coadjoint representations are not equivalent.
The coadjoint action can be worked out using the above inner product:
\begin{equation}
  \label{eq:coadjoint-def}
  \bigl< \Ad^*_{(\sigma, A, v, h)}(\mu, Y, p, \varepsilon),
    (\lambda, X, q, \theta)\bigr> = \left<(\mu, Y, p, \varepsilon), \Ad_{(\sigma, A, v, h)^{-1}} (\lambda, X, q, \theta) \right>.
\end{equation}
Using the explicit expression \eqref{eq:adjoint-action} for the
adjoint action, we find that the coadjoint action is given by
\begin{equation}
  \label{eq:coadjoint-action}
  \Ad^*_{(\sigma, A, v, h)} (\mu, Y, p, \varepsilon) = (\mu', Y', p', \varepsilon'),
\end{equation}
where
\begin{equation}
  \label{eq:coadjoint-action-too}  
  \begin{split}
    \mu' &= \mu + w \sigma^{-w} v^T A p + z \sigma^{-z} h \varepsilon\\
    Y' &= A Y A^{-1} +  \sigma^{-w} (A p v^T - v (A p)^T)\\
    p' &= \sigma^{-w} A  p\\
    \varepsilon' &= \sigma^{-z} \varepsilon,
  \end{split}
\end{equation}
which differs from the adjoint action \eqref{eq:adjoint-action}, as
expected. This expression is to be interpreted as the action of the
group $G = \RR^+ \times \SO(d) \times \RR^d \times \RR$ (with the
group multiplication \eqref{eq:G-mult}) on the vector space
$\RR \oplus \wedge^2 \RR^d \oplus \RR^d \oplus \RR$. Notice,
parenthetically, that for nonzero $\varepsilon$ and $p$ the rational
function 
$\frac{\varepsilon^{2/z}}{(p^Tp)^{1/w}}$ is an invariant of the
coadjoint orbit. To interpret this invariant we restrict to $w=1$ and
rewrite it as $\varepsilon^{2} = v^{2}_{z} (p^{T}p)^{z}$ which we can
understand as a dispersion relation. For $z=1$ it indeed agrees with
the well known relation $\varepsilon^{2} = c^{2} \, p^{T}p$ of
massless Poincaré particles where $v_{1}$ is given by the speed of
light $c$. This also agrees with the Lifshitz particle presented in
\cite{QuimLifshitzParticle}. For vanishing $p$ the spin $\tr(Y^{T}Y)$
is an invariant and when additionally $\varepsilon$ is zero $\mu$ is
also invariant.

Special cases of the coadjoint action are
\begin{equation}
  \label{eq:special-coad}
  \begin{split}
   \Ad^*_{(\sigma, \id, 0,0)} (\mu, Y, p, \varepsilon) &= (\mu, Y, \sigma^{-w}p,\sigma^{-z}\varepsilon )\\
   \Ad^*_{(1, A, 0,0)} (\mu, Y, p, \varepsilon) &= (\mu, A Y A^{-1}, A p, \varepsilon)\\
   \Ad^*_{(1, \id, v,0)} (\mu, Y, p, \varepsilon) &= (\mu + w v^T p, Y + p v^T- v p^T, p, \varepsilon)\\
   \Ad^*_{(1, \id, 0,h)} (\mu, Y, p, \varepsilon) &= (\mu + z h \varepsilon, Y,
   p, \varepsilon).
  \end{split}
\end{equation}
The first two are as expected: $(\sigma, \id, 0,0)$ rescales $p$ and
$\varepsilon$, whereas $(1,A,0,0)$ acts like a rotation: conjugating
$Y$ and rotating the vector $p$.

\subsubsection{Structure of coadjoint orbits}
\label{sec:struct-coadj-orbits}

The group $G$ is actually a semidirect product $\CO(d) \ltimes T$,
where the abelian group $T \cong \RR^{d+1}$ and where the action of
$(\sigma,A) \in \CO(d)$ on $(v,h) \in T$ is given by
\begin{equation}
  (\sigma, A) \cdot (v,h) = (\sigma^w A v, \sigma^z h).
\end{equation}
Coadjoint orbits of such semidirect products have been studied, for
example, in the thesis of Oblak \cite{Oblak:2016eij} and the
references therein. Let us write $G = K \ltimes T$, with $K$ a
connected Lie group and $T$ abelian. We will let $\g = \k \oplus \t$
as a vector space and hence $\g^* = \k^* \oplus \t^*$, where we
identify $\k^*$ with the annihilator
$\t^o = \left\{\alpha \in \g^* \middle | \alpha(X) = 0~\forall X \in
  \t\right\}$ of $\t$ and, similarly, $\t^*$ with the annihilator
$\k^o$ of $\k$.  Let $\alpha \in \g^*$ and decompose it as
$\alpha = (\kappa,\tau) \in \k^* \oplus \t^*$.  Since $K$ acts on $T$
by automorphisms, it acts on its Lie algebra $\t$ and hence on the
dual $\t^*$.  Let $\eO_\tau \subset \t^*$ denote the $K$-orbit of
$\tau$ in $\t^*$.  There is a $K$-equivariant diffeomorphism
$\eO_\tau \cong K/K_\tau$, where
$K_\tau = \left\{ k \in K \middle | k \cdot \tau = \tau\right\}$ is
the stabiliser of $\tau$ in $K$. This exhibits $K$ as the total space
of a principal $K_\tau$ bundle $K \to \eO_\tau$ and given any manifold
$M$ on which $K_\tau$ acts, we may construct an associated fibre
bundle $K \times_{K_\tau} M \to \eO_\tau$, whose typical fibre is a
copy of $M$. For example, $K_\tau$ acts on $\k^*$ and, since
$\k_\tau \subset\k$ is a Lie subalgebra, this action preserves the
annihilator $\k_\tau^o$.  Since this is a linear representation of
$K_\tau$, the associated fibre bundle $K \times_{K_\tau} \k_\tau^o$ is
a vector bundle, and using the isomorphism
$\k_\tau^o \cong (\k/\k_\tau)^*$ can be seen to be the cotangent
bundle $T^*\eO_\tau$ of $\eO_\tau$.  Another example of associated
fibre bundles, this time not a vector bundle, is given by considering
a coadjoint orbit $\eO'$ of $K_\tau$ and constructing
$K \times_{K_\tau} \eO'$.  With these definitions behind us, we can
describe the coadjoint orbits of $G = K \ltimes T$.  The $G$-coadjoint
orbit of $\alpha = (\kappa, \tau)$ is the associated fibre bundle
$K \times_{K_\tau} (\k_\tau^o \times \eO_{\kappa_\tau}) \to \eO_\tau$,
where $\kappa_\tau \in \k_\tau^*$ is the restriction of $\kappa$ to
$\k_\tau$ and $\eO_{\kappa_\tau}$ is its $K_\tau$-coadjoint orbit.
The total space of the bundle
$K \times_{K_\tau} (\k_\tau^o \times \eO_{\kappa_\tau}) \to \eO_\tau$
is the fibred product of the cotangent bundle $T^*\eO_\tau$ and the
associated fibre bundle $K \times_{K_\tau} \eO_{\kappa_\tau}$ over
$\eO_\tau$.  As a check, notice that the dimension is $2 \dim \eO_\tau
+ \dim \eO_{\kappa_\tau}$, which is indeed even, since
$\eO_{\kappa_\tau}$ is a coadjoint orbit itself.

Two extremal cases are worth noting: if $\tau = 0$, then this simply
the $K$-coadjoint orbit of $\kappa$, whereas if $\kappa = 0$, this is
simply the cotangent bundle $T^*\eO_\tau$.

To determine the coadjoint orbits of our groups of interest
$G = \CO(d) \ltimes T$ we need to first decompose $\t^*$ into
$\CO(d)$-orbits and determine their stabilisers and then to determine
the coadjoint orbits of the stabilisers.  The action of $\CO(d)$ on
$\t^*$ can be read off from the last two entries in
equation~\eqref{eq:coadjoint-action} for the coadjoint action after
setting $v$ and $h$ to zero:
\begin{equation}
  \label{eq:co-action-on-t-star}
  (\sigma,A) \cdot (p, \varepsilon) = (\sigma^{-w} A p, \sigma^{-z} \varepsilon). 
\end{equation}

To continue, we must consider several cases depending on the values of
$(z,w)$.

\subsubsection{$\CO(d)$-orbits in $\t^*$ for $z=w=0$}
\label{sec:cod-orbits-w0z0}

If $z=w=0$ the action~\eqref{eq:co-action-on-t-star} reduces to
\begin{equation}
  (\sigma,A) \cdot (p, \varepsilon) = (A p, \varepsilon).
\end{equation}
There are two kinds of orbits:
\begin{itemize}
\item point-like orbits $\{(0,\varepsilon)\}$, with stabiliser $\CO(d)$; and
\item spherical orbits $S^{d-1}_{|q|} \times \{\varepsilon\}$ through
  $(p\neq 0,\varepsilon)$, with stabiliser
  $\CO(p^\perp) = \RR^+ \times \SO(p^\perp)
  \cong \CO(d-1)$.  Here the notation is that $S^{d-1}_{|p|}$ is the
  sphere of radius $|p|$, the euclidean norm of $p \in \RR^d$.
\end{itemize}

\subsubsection{$\CO(d)$-orbits in $\t^*$ for $w=0$ and $z \neq 0$}
\label{sec:cod-orbits-w0z}

We keep $w=0$ but now have $z \neq 0$, which can be set to
$z=1$ without loss of generality by rescaling $D$.  Then the
action~\eqref{eq:co-action-on-t-star} reduces to,
\begin{equation}
  (\sigma,A) \cdot (p, \varepsilon) = (A p, \sigma^{-1} \varepsilon).
\end{equation}
We have the following orbits, depending on $(p,\varepsilon)$:
\begin{itemize}
\item a point-like orbit $\{(0,0)\}$, with stabiliser $\CO(d)$;
\item a spherical orbit $S^{d-1}_{|p|} \times \{0\}$ through
  $(p\neq 0,0)$, with stabiliser $\CO(p^\perp)$;
\item two ray-like orbits $\{0\} \times \RR^\pm$ through $(0,\varepsilon)$
  with $\pm \varepsilon > 0$ and stabilisers $\SO(d)$;
\item two cylindrical orbits $S^{d-1}_{|p|} \times \RR^\pm$ through
  $(p\neq 0, \varepsilon)$, with $\pm \varepsilon > 0$, and stabiliser
  $\SO(p^\perp)$.
\end{itemize}

\subsubsection{$\CO(d)$-orbits in $t^*$ for $w\neq 0$ and $z =0$}
\label{sec:cod-orbits-w1z0}

Next we consider $z=0$ and $w\neq 0$.  Again we can set $w=1$ without
loss of generality, resulting in the action
\begin{equation}
  (\sigma,A) \cdot (p, \varepsilon) = (\sigma^{-1} A p, \varepsilon).
\end{equation}
There are two kinds of orbits:
\begin{itemize}
\item point-like orbits $\{(0,\theta)\}$, with stabiliser $\CO(d)$; and
\item orbits $\left( \RR^d\setminus\{0\}\right) \times \{\theta\}$,
  through $(p\neq 0,\varepsilon)$, with stabiliser $\SO(p^\perp)$.
\end{itemize}

\subsubsection{$\CO(d)$-orbits in $\t^*$ for $w\neq 0$ and $z \neq 0$}
\label{sec:cod-orbits-w1z}

Finally we have the case $w = 1$ (without loss of generality) and $z
\neq 0$, resulting in the action
\begin{equation}
  (\sigma,A) \cdot (p, \varepsilon) = (\sigma^{-1} A p, \sigma^{-z} \varepsilon).
\end{equation}
We have the following orbits:
\begin{itemize}
\item a point-like orbit $\{(0,0)\}$, with stabiliser $\CO(d)$;
\item two ray-like orbits $\{0\} \times \RR^\pm$ through $(0,\varepsilon)$
  with $\pm \varepsilon > 0$ and stabiliser $\SO(d)$;
\item an orbit $\left( \RR^d\setminus\{0\}\right) \times \{0\}$,
  through $(p\neq 0,0)$, with stabiliser $\SO(p^\perp)$; and
\item cylindrical orbits through $(p \neq 0, \varepsilon)$ with $\pm \varepsilon >
  0$, and stabiliser $\SO(p^\perp)$.  These orbits can be thought of
  as a sphere-bundle over the half-line, where the radius of the
  sphere varies with the point on the line.
\end{itemize}

\subsubsection{Summary}
\label{sec:summary-1}

In summary, the stabilisers are in all cases isomorphic to one of $\CO(d)$,
$\CO(d-1)$, $\SO(d)$ or $\SO(d-1)$.  We next determine the coadjoint
orbits of these groups, where we restrict ourselves to $d\geq 2$.  The
coadjoint action of $\CO(d)$ on $\co(d)^*$ can be read off from the
first two entries in equation~\eqref{eq:coadjoint-action} after
setting $v$ and $h$ to zero:
\begin{equation}
  \Ad^*_{(\sigma, A)} (\mu, Y) = (\mu, A Y A^{-1}).
\end{equation}
We therefore see that the coadjoint orbit through $(\mu,Y)$ is
$\{\mu\} \times \eO_Y$, where $\eO_Y$ is the coadjoint orbit of
$Y \in \so(d)^*$ under $\SO(d)$.  In other words, we are left with the
task of studying the coadjoint orbits of $\SO(n)$ for $n\geq 1$, since
$n=d$ or $n=d-1$ and $d \geq 2$.  Of course, the cases $n=1$ and $n=2$
have only point-like orbits: this is because $\so(1)= 0$ and $\so(2)$
is abelian.  The coadjoint orbits for $\so(3)$ are well known: we have
the origin of $\so(3)^*$ and then the spheres of radius equal to the
norm of $Y$ under the euclidean inner product induced by the Killing
form.  What about for $n>3$?  Being semisimple, the adjoint and
coadjoint representations are equivalent, and hence we may work with
the adjoint orbits.  These have been characterised in
\cite{arathoon2018hermitian}, where it is shown in Theorem~3.1 of that
paper, that coadjoint orbits of $SO(n)$ are hermitian flag manifolds
in $\RR^n$.  We describe some of them in
Appendix~\ref{sec:coadjoint-orbits-son}.

This describes all the ingredients required to determine, at least in
principle, the coadjoint orbits of the Lifshitz groups associated to
$\a_1$, $\a_2$ and $\a_3^z$.

\subsubsection{Coadjoint orbits from Lifshitz geodesics}
\label{sec:coadj-orbits-lifsh-4}

Let us consider the case $\g = \a_3^z$, thought of as Killing vector
fields \eqref{eq:lifshitz-kvfs} in the Lifshitz spacetime
$(M^{d+2},g)$ for the metric $g$ given in
equation~\eqref{eq:lifshitz}.\footnote{There are other
  Lifshitz-invariant metrics, but since this section is for the
  purpose of illustrating the method, we pick the metric which was
  already discussed above.}  Let $\gamma(s) = (t(s),r(s),x^a(s))$ be
an affinely parametrised geodesic.  As discussed in
Section~\ref{sec:coadjoint-orbits}, $\gamma$ defines an element
$\alpha_\gamma = (\Delta, \ell, k, E) \in \g^*$, where
\begin{equation}
  \begin{split}
    k_a &= g(\dot\gamma, \xi_{P_a}) = \tfrac{\dot x_a}{r^2}\\
    \ell_{ab} &= g(\dot\gamma, \xi_{J_{ab}}) = - x_a k_b + x_b k_a = \tfrac{-x_a \dot x_b + x_b \dot x_a}{r^2}\\
    E &= g(\dot\gamma, \xi_H) = - \tfrac{\dot t}{r^{2z}}\\
    \Delta &= g(\dot\gamma, \xi_D) = \tfrac{\dot r}{r} + z \tfrac{t\dot t}{r^{2z}} + \tfrac{x_a \dot x^a}{r^2} = \tfrac{\dot r}{r} - z t E + k_a x^a.
  \end{split}
\end{equation}
Under the coadjoint action $\Ad^*_{(\sigma,A,v,h)} (\Delta, \ell, k,
E) = (\Delta', \ell', k', E')$, where
\begin{equation}
  \begin{split}
    \Delta' &= \Delta + \sigma^{-1} v^T A k + z \sigma^{-z} E h\\
    \ell'_{ab} &= (A \ell A^{-1})_{ab} + \sigma^{-1} \left( (Ak)_a v_b - (Ak)_b v_a \right)\\
    k'_a &= \sigma^{-1} (Ak)_a\\
    E' &= \sigma^{-z} E.
  \end{split}
\end{equation}
Consider a geodesic with $\ell = 0$ and $\Delta = 0$ and with $E \neq
0$ and $k \neq 0$.  Then the corresponding coadjoint orbit has
dimension $2d$ and, from the discussion in
Section~\ref{sec:struct-coadj-orbits}, it is the cotangent bundle
$T^*\eO_\alpha$, where $\eO_\alpha$ is the orbit of $\alpha = k_a \pi^a
+ E \eta \in \g^*$ under the action of the subgroup $\CO(d)$ generated
by $J_{ab}$ and $D$.  This orbit is a generalised cylinder with
equation $E^2/|k|^{2z} = c$ for some constant $c>0$.

The stabiliser subgroup of $\alpha = k_a \pi^a + E \eta$ is isomorphic
to $\SO(d-1) \times \RR$, where $\SO(d-1)$ is the subgroup of $\SO(d)$
which fixes $k_a \pi^a$ and $\RR$ is the subgroup $\Gamma \subset G$
consisting of elements of the form
$(1,1, -\tfrac{z E h}{|k|^2} k, h )$ for $h \in \RR$.  The coadjoint
orbit of $\alpha$ is the base of a (trivial) principal $\Gamma$-bundle
whose total space is the evolution space in the sense of Souriau
\cite{MR1461545}.  It has a presymplectic structure (i.e., a closed
2-form) obtained by pulling back the Kirillov--Kostant--Souriau
symplectic form on the coadjoint orbit, whose kernel defines a
one-dimensional (integrable) distribution on the evolution space,
whose leaves are the particle trajectories.

It is possible to define a particle lagrangian for such trajectories
purely from the data defining the coadjoint orbit.  One might argue
that we already have a lagrangian, namely the one for geodesics.  Let
us make two remarks about this.  The first is that extremals of the
geodesic action principle are affinely parametrised geodesics
regardless of the causal type, whereas the action constructed from the
coadjoint orbit is tied to a causal type.  This is well-known from the
case of Minkowski geodesics, since the coadjoint orbits of lightlike
and timelike geodesics are different.  The second remark is that we do
not always have an invariant metric on a homogeneous space and hence
we do not necessarily have an action principle for geodesics, whereas
we can often construct an action principle from the coadjoint orbit.

Let us illustrate this method for the coadjoint orbit of $(0,0,k,E)$,
with $k\in\RR^d$ and $E\in\RR$ nonzero.  As we saw above the
stabiliser subalgebra of $(0,0,k,E)$ is isomorphic to $\so(d-1) \oplus
\RR Z$ where $\so(d-1)$ is the stabiliser of $k$ in the
vector representation of $\so(d)$ and $Z = (0,0, -\tfrac{z E}{|k|^2}
k, 1 ) \in \g$.  The corresponding evolution space is also a
homogeneous space of $G$ with stabiliser $\SO(d-1)$, the subgroup of
$\SO(d)$ which fixes $k$.  We may parametrise the evolution space
locally via the coset representative
\begin{equation}
  g = \underbrace{e^{t H} e^{x^a P_a}}_{g_0} e^{r D} e^{\theta^i R_i},
\end{equation}
where $R_i$, $i=1,\dots,d-1$ generate rotations which do not preserve
$k$ and where $g_0$ parametrises a point in the $(d+1)$-dimensional
Lifshitz--Weyl spacetime obtained from $M$ by quotienting by the
one-parameter group generated by $D$.  This is just for convenience in
the calculation of the pull-back of the left-invariant Maurer--Cartan
one-form, which gives
\begin{equation}
  g^{-1}dg = dr D  + e^{-zr} dt H + e^{-r} \left( \cos\|\theta\|
    dx^\parallel + \frac{\sin\|\theta\|}{\|\theta\|} \theta \cdot
    x^\perp \right) P^\parallel + \cdots
\end{equation}
where $x^\parallel = \frac{x \cdot k}{\|k\|^2}k$ is the component of
$x$ along $k$ (and similarly for $P^\parallel$),
$x^\perp = x - x^\parallel$ is the component perpendicular to $k$ and
$\|\theta\| = \delta_{ij}\theta^i\theta^j$.  In this expression we
have omitted any terms which are not invariant under $\so(d-1)$.  If
$\gamma(s)$ is a curve in the evolution space, it is described in
these coordinates by $(t(s),x^a(s),r(s),\theta^i(s))$ and the
lagrangian is the pull-back of a linear combination of the
$\so(d-1)$-invariant components of $g^{-1}dg$ to the $s$-interval
parametrising the curve.  Letting dots denote derivative with respect
to $s$, we have that the most general lagrangian is given by
\begin{equation}
  L = c_0 \dot r + c_1 e^{-zr} \dot t + c_2 e^{-r} \left(
    \cos\|\theta\| \dot x^\parallel +
    \frac{\sin\|\theta\|}{\|\theta\|} \theta \cdot \dot x^\perp\right),
\end{equation}
for some constants $c_0,c_1,c_2$.  We can ignore the first term, since
it is a total derivative, so without loss of generality we may set
$c_0 = 0$.

The canonical momenta are given by
\begin{equation}
  \begin{split}
    E &:= \frac{\partial L}{\partial \dot t} = c_1 e^{-zr}\\
    p^\parallel &:= \frac{\partial L}{\partial \dot x^\parallel} = e^{-r} c_2 \cos\|\theta\|\\
    p^\perp_i &:= \frac{\partial L}{\partial \dot x_i^\perp} = e^{-r} c_2 \frac{\sin\|\theta\|}{\|\theta\|} \theta_i.
  \end{split}
\end{equation}
The Euler--Lagrange equations say that $E, p^\parallel, p^\perp_i$ are
constant, so that if we assume that none of $z, c_1, c_2$ are zero, we
obtain that $r$ and $\theta_i$ are constant.  The momentum satisfies the
constraint
\begin{equation}
  p^2 = (p^\parallel)^2 + (p^\perp)^2 = c_2^2 e^{-2r}
\end{equation}
so that we recover the constraint that is satisfied by the coadjoint
orbit: namely, $E^2/\|p\|^{2z} = (c_1/c_2^z)^2$.  For prior work
related to dynamical realisations of the Lifshitz group see the
unpublished work of Gomis and Kamimura described in
\cite{QuimLifshitzParticle} and the more recent work of Galajinsky 
\cite{Galajinsky:2022bwq}.

\subsubsection{Central extensions}
\label{sec:central-extensions-1}

As shown in Table~\ref{tab:lifshitzcent}, for $d=2$ the Lie algebras
$\a_1$ and $\a_2$ admit central extensions which integrate to central
extensions of the corresponding Lie groups. Coadjoint orbits of these
central extensions give rise to homogeneous symplectic manifolds of
the corresponding two-dimensional Lifshitz Lie groups. We shall not
discuss them in this paper, but might it be interesting to study them
in the future.

The Lie algebras $\a_1$ and $\a_3^{z=0}$ do have a central extension
for all values of $d$.  In the case of $\a_1$, the central extension
is isomorphic to $\iso(d) \oplus \h_3$, where $\h_3$ is a Heisenberg
algebra.  Coadjoint orbits of a direct product of Lie groups are
products of coadjoint orbits: those of the Heisenberg group are discussed in
Appendix~\ref{sec:coadj-orbits-heis}, whereas those of the euclidean
group $\ISO(d) = \SO(d) \ltimes \RR^d$ can be obtained
via the method explained in Section~\ref{sec:struct-coadj-orbits}: they
boil down to the determination of the coadjoint orbits of $\SO(d-1)$,
being the stabiliser of a nonzero $p \in (\RR^d)^*$.  The central
extension of $\a_3^{z=0}$ is now isomorphic to a semidirect product
$(\so(d) \oplus \h_3) \ltimes \RR^d$, where  $\RR^d =
\left<P_a\right>$ is an abelian ideal.  The discussion in
Section~\ref{sec:struct-coadj-orbits} again applies and all the
situation is very similar to the one described above with $\SO(d)
\times H_3$ replacing $\CO(d)$. We do not discuss them further here,
but leave them for future work.

\subsection{$\a_4^\pm$ and $\a_5^\pm$}
\label{sec:a_4pm-a_5pm}

We discuss these two Lie algebras together by introducing a parameter
$z \in \{0,1\}$ and letting $[D,H]=z H$.  The Lie algebra $\a_4^\pm$
corresponds to $z=0$ and $\a_5^\pm$ to $z=1$.  These Lie algebras are
direct sums of the Lie algebras: the subalgebra spanned by
$J_{ab},P_a$ and the two-dimensional Lie algebra spanned by $D,H$. 
The former Lie algebra has brackets $[P_a,P_b]= \pm J_{ab}$ in
addition to those involving $J_{ab}$.  It is isomorphic to $\so(d,1)$
if the sign is $+$ and to $\so(d+1)$ if the sign is $-$.  A coadjoint
orbit of either of these two Lifshitz Lie groups is therefore a
product of a coadjoint orbit of $\SO(d,1)_0$ or $\SO(d+1)$, for
$d \geq 2$ and a coadjoint orbit of the two-dimensional Lie group
generated by $D$ and $H$.  They are described in the appendices.

The connected (and simply-connected) Lie group $G$ generated by $D$ and
$H$ is given by
\begin{equation}
  G = \left\{
    \begin{pmatrix}
      a & b \\ 0 & a^{1-z}
    \end{pmatrix}
\middle | a,b \in \RR,~ a> 0\right\}.
\end{equation}
Let us consider $\a_5^\pm$, so $z=1$.  As shown in
Section~\ref{sec:centr-extens-a_4}, this Lie algebra admits no central
extensions and hence any simply-connected homogeneous symplectic
manifold covers a coadjoint orbit.  These are determined in
Appendix~\ref{sec:coadjoint-orbits-G2}.

As shown in Section~\ref{sec:centr-extens-a_4}, the Lie algebra
$\a_4^\pm$ admits a one-dimensional central extension $[D,H]=Z$.  This
promotes the two-dimensional abelian group generated by $D,H$ to the
Heisenberg group, whose coadjoint orbits are described in
Appendix~\ref{sec:coadj-orbits-heis}.  The relevant coadjoint orbits
of the central extension are now products of coadjoint orbits of the
Heisenberg group with those of either $\SO(d+1)$ or $\SO(d,1)_0$.

\subsection{$\a_6$ and $\a_7$}
\label{sec:a_6-a_7}

As in Section~\ref{sec:centr-extens-a_6}, we introduce a parameter $w
\in \{0,1\}$ in order to treat both algebras simultaneously.  Let $\g$
(depending on $w$) be the Lie algebra under consideration.  The brackets are then
\begin{equation}
  [D,P_a] = w P_a, \qquad [D,H] = 2 w H, \qquad [J,P_a] =
  \epsilon_{ab} P_b \qquad\text{and}\qquad [P_a,P_b]=\epsilon_{ab} H.
\end{equation}
It has the structure of a semidirect product of the abelian subalgebra
spanned by $J$ and $D$ acting as derivations on the Heisenberg ideal
spanned by $P_a,H$.

The Lie algebra $\g = \k \ltimes \h_3$ is a semidirect product of
the abelian two-dimensional Lie algebra $\k = \left<J,D\right> \cong
\co(2)$ with the Heisenberg algebra: $\h_3 = \left<P_a, H\right>$.

\subsubsection{Group law and (co)adjoint actions}
\label{sec:group-law-coadjoint}

The first task is to explicitly write down the group law in $G = K
\ltimes H_3$.  Let $H_3$ be the Heisenberg group generated by $P_a,H$.
It is a unipotent matrix group diffeomorphic to $\RR^3$, given
explicitly by
\begin{equation}
  H_3 = \left\{\begin{pmatrix} 1 & a & c \\ 0 & 1 & b \\ 0 & 0 & 1 \end{pmatrix} \middle | a,b,c \in \RR \right\}.
\end{equation}
Every element of $H_3$ can be uniquely written as a product of matrix
exponentials:
\begin{equation}
  \begin{pmatrix} 1 & a & c \\ 0 & 1 & b \\ 0 & 0 & 1 \end{pmatrix} =
  \exp
  \begin{pmatrix}
    0 & 0 & c \\ 0 & 0 & 0 \\ 0 & 0 & 0
  \end{pmatrix}
  \exp
  \begin{pmatrix}
    0 & 0 & 0 \\ 0 & 0 & b \\ 0 & 0 & 0
  \end{pmatrix}
  \exp
  \begin{pmatrix}
    0 & a & 0 \\ 0 & 0 & 0 \\ 0 & 0 & 0
  \end{pmatrix} = \exp(c H) \exp(b P_2) \exp(a P_1),
\end{equation}
where
\begin{equation}
  P_1 = \begin{pmatrix}
    0 & 1 & 0 \\ 0 & 0 & 0 \\ 0 & 0 & 0
  \end{pmatrix}, \qquad
  P_2 = \begin{pmatrix}
    0 & 0 & 0 \\ 0 & 0 & 1 \\ 0 & 0 & 0
  \end{pmatrix} \qquad\text{and}\qquad
  H = \begin{pmatrix}
    0 & 0 & 1 \\ 0 & 0 & 0 \\ 0 & 0 & 0
  \end{pmatrix}.
\end{equation}
It is easy to write down the action of $K$ on the Lie algebra $\h_3$
spanned by $P_a, H$ by exponentiating $[J,P_a] = \epsilon_{ab} P_b$,
$[D,P_a] = w P_a$ and $[D,H]=2w H$.  If we write $P = P_1 +
i P_2$ and then the adjoint action of $J$ is simply multiplication by
$-i$.  Therefore we find that, of course, $\exp(\theta J)H=H$ and that
\begin{equation}
  \exp(\theta J) (P_1 + i P_2) = e^{-i \theta}(P_1 + i P_2) = (P_1 \cos\theta + P_2
  \sin\theta) + i (P_2 \cos\theta- P_1 \sin\theta),
\end{equation}
whereas
\begin{equation}
  \exp(\sigma D) P_a = e^{w \sigma} P_a \qquad\text{and}\qquad \exp(\sigma D) H =
  e^{2 w \sigma} H.
\end{equation}
The general element $k = \exp(\sigma D) \exp (\theta J) \in K$ acts on $\h_3$ as
\begin{equation}\label{eq:auto-phi}
  \begin{split}
    g \cdot P_1 &= e^{w \sigma} (P_1 \cos \theta + P_2 \sin \theta) \\
    g \cdot P_2 &= e^{w \sigma} (P_2 \cos \theta - P_1 \sin \theta) \\
    g \cdot H &= e^{2 w \sigma} H.
  \end{split}
\end{equation}
One checks that this is an automorphism of $\h_3$:
\begin{equation}
  [g \cdot X, g \cdot Y] = g \cdot [X,Y] \qquad \forall X,Y \in \h_2,
\end{equation}
so that it defines a map $\phi : K \to \Aut(\h_3)$ sending every
$g \in K$ to the automorphism $\phi_g$ of $\h_3$ defined by
$\phi_g X = g \cdot X$ in equation~\eqref{eq:auto-phi}. Fix $g \in
K$. Then $\phi_g: \h_3 \to \h_3$ is in particular a Lie algebra
homomorphism and thus, by the Lie correspondence, lifts to a unique
Lie group homomorphism $\Phi_g : H_3 \to H_3$. To work it out we argue
as follows. The graph $\gamma \subset \h_3 \oplus \h_3$ of $\phi_g$ is
a Lie subalgebra of $\h_3 \oplus \h_3$ (because $\phi_g$ is a
homomorphism) and hence it exponentiates there to a connected subgroup
$\Gamma \subset H_3 \times H_3$.  There are two cartesian projections
$H_3 \times H_3 \to H_3$: restricting the left projection to $\Gamma$
gives a covering $\Gamma \to H_3$ which can be shown in this case to
be an isomorphism, whereas restricting the right projection to
$\Gamma$ gives the desired $\Phi_g$.  In detail, the graph
$\gamma \subset \h_3 \oplus \h_3$ is spanned by
 \begin{equation}
   \begin{split}
     (P_1, g \cdot P_1) &= (P_1, e^{w\sigma} (P_1 \cos \theta + P_2 \sin \theta))\\
     (P_2, g \cdot P_2) &= (P_2, e^{w\sigma} (P_2 \cos \theta - P_1 \sin \theta))\\
     (H, g \cdot H) &= (H, e^{2 w\sigma} H).
   \end{split}
 \end{equation}
 We can write $\gamma$ explicitly as the span of the following three
 block-diagonal $6 \times 6$ matrices:
 \begin{equation}
   \begin{split}
     \widehat P_1 &:=
     \begin{pmatrix}
       0 & 1 & 0 & & & \\ 0 & 0 & 0 & & & \\ 0 & 0 & 0 & & & \\
       & & & 0 & e^{w\sigma} \cos \theta & 0 \\ & & & 0 & 0 & e^{w \sigma} \sin \theta \\ & & & 0 & 0 & 0
     \end{pmatrix}\\
     \widehat P_2 &:=
     \begin{pmatrix}
       0 & 0 & 0 & & & \\ 0 & 0 & 1 & & & \\ 0 & 0 & 0 & & & \\
       & & & 0 & - e^{w\sigma} \sin \theta & 0 \\ & & & 0 & 0 & e^{w \sigma} \cos \theta \\ & & & 0 & 0 & 0
     \end{pmatrix}\\
     \widehat H &:=
     \begin{pmatrix}
       0 & 0 & 1 & & & \\ 0 & 0 & 0 & & & \\ 0 & 0 & 0 & & & \\
       & & & 0 & 0 & e^{2 w \sigma} \\ & & & 0 & 0 & 0 \\  & & & 0 & 0 & 0
     \end{pmatrix}\\
   \end{split}
 \end{equation}
and they exponentiate to the subgroup of $\GL(6,\RR)$ consisting of
matrices of the form
\begin{multline}
  \exp(h \widehat H) \exp( p_2 \widehat P_2) \exp (p_1 \widehat P_1)\\
  = 
  \begin{pmatrix}
    1 & p_1 & h & & & \\  0 & 1 & p_2 & & & \\  0 & 0 & 1 & & & \\
    & & & 1 & e^{w\sigma} (p_1 \cos \theta - p_2 \sin \theta) & e^{2 w\sigma} (h - p_1 p_2
    \sin^2 \theta + \tfrac12 (p_1^2 - p_2^2) \sin \theta \cos \theta) \\
    & & & 0 & 1 & e^{w\sigma} (p_2 \cos \theta + p_1 \sin \theta) \\
    & & & 0 & 0 & 1
  \end{pmatrix}
\end{multline}
from where read off that the Lie group automorphism $\Phi_g : H_3 \to
H_3$ corresponds to
\begin{equation}
  \begin{pmatrix}
    1 & p_1 & h \\ 0 & 1 & p_2 \\ 0 & 0 & 1 
  \end{pmatrix} \mapsto
  \begin{pmatrix}
    1 & e^{w\sigma} (p_1 \cos \theta - p_2 \sin \theta) & e^{2 w\sigma} (h - p_1 p_2
    \sin^2 \theta + \tfrac12 (p_1^2 - p_2^2) \sin \theta \cos \theta) \\
    0 & 1 & e^{w\sigma} (p_2 \cos \theta + p_1 \sin \theta) \\
    0 & 0 & 1
  \end{pmatrix}.
\end{equation}
One checks that $\Phi : K \to \Aut(H_3)$ is indeed a Lie group
homomorphism, so that $\Phi_{g g'} = \Phi_g \Phi_{g'}$ for all $g,g'
\in K$.

With these results in hand, we can now write down the group law on
$G = K \ltimes H_3$. If
$(k_1, h_1), (k_2,h_2) \in G$, their product is given by
\begin{equation}
  (k_1, h_1) (k_2, h_2) = (k_1 k_2, h_1 (\Phi_{k_1}h_2)),
\end{equation}
from where we read off that the inverse of $(k,h) \in G$ is
given by
\begin{equation}
  (k,h)^{-1} = (k^{-1}, \Phi_{k^{-1}}h^{-1}).
\end{equation}
Notice that any automorphism commutes with inversion so that
$(\Phi_kh)^{-1} = \Phi_kh^{-1}$.  Explicitly, the inverse of the group
element with coordinates $(\theta,\sigma,p_1,p_2,h)$ is the group
element with
coordinates
\begin{equation}\small
  (-\theta, -\sigma, -e^{-w \sigma} (p_1 \cos \theta + p_2 \sin \theta), e^{-w
    \sigma}(p_1\sin\theta - p_2 \cos \theta), e^{-2 w \sigma} (p_1p_2\cos^2 \theta - h -
  \tfrac12 (p_1^2-p_2^2)\sin \theta \cos \theta)).
\end{equation}

We can now work out the adjoint representation of $G$.  If
$(k,h) \in G$ and $(k_2(t),h_2(t))$ is a curve in
$G$ with $k_2(0) = 1_{K}$, $h_2(0) = 1_{H_3}$, $k_2'(0) =
X \in \k$ and $h_2'(0)=Y \in \h_3$, then
\begin{equation}
  \begin{split}
    \widehat \Ad_{(k,h)} (X,Y) &= \left. \frac{d}{dt}\right|_{t=0} (k,h) (k_2(t),h_2(t))(k^{-1},\Phi_{k^{-1}} h^{-1})\\
    &= \left. \frac{d}{dt}\right|_{t=0}   (k k_2(t), h  (\Phi_{k}h_2(t))) (k^{-1},\Phi_{k^{-1}} h^{-1}) \\
    &= \left. \frac{d}{dt}\right|_{t=0}  (k k_2(t) k^{-1}, h  (\Phi_{k}h_2(t)) \Phi_{k k_2(t)}(\Phi_{k^{-1}}h^{-1}))\\
    &=  \left(\Ad^{K}_{k} X, \left. \frac{d}{dt}\right|_{t=0} \left( h (\Phi_{k}h_2(t)) \Phi_{k k_2(t)k^{-1}}h^{-1}\right) \right)\\
    &=  \left(X, \left. \frac{d}{dt}\right|_{t=0} \left( h (\Phi_{k}h_2(t)) \Phi_{k k_2(t)k^{-1}}h^{-1}\right) \right),
  \end{split}
\end{equation}
using that $K$ is abelian.

If we let $k = (\theta,\sigma) \in K$ and $X = (x,y) \in
\k$, and also
\begin{equation}\label{eq:hX-defs}
  h =
  \begin{pmatrix}
    1 & p_1 & h \\ 0 & 1 & p_2 \\ 0 & 0 & 1
  \end{pmatrix} \in H_3
  \qquad\text{and}\qquad
  Y =
  \begin{pmatrix}
    0 & u & s \\ 0 & 0 & v\\ 0 & 0 & 0
  \end{pmatrix} \in \h_3,
\end{equation}
then $\widehat\Ad_{(k,h)}(X,Y) = (X,Y')$, where
\begin{equation}
  Y' =
  \begin{pmatrix}
    0 & -w y p_1 + x p_2 + e^{w \sigma} (u \cos \theta - v \sin \theta) & \Xi\\
    0 & 0 & -w y p_2 - p_1 x + e^{w\sigma} (u \sin\theta + v \cos \theta)\\
    0 & 0 & 0
  \end{pmatrix},
\end{equation}
where
\begin{equation}
  \Xi = w y (p_1p_2 - 2 h) - \tfrac12 (p_1^2+p_2^2) x + e^{2 w \sigma} s +
  e^{w\sigma} ((u p_1 + v p_2) \sin \theta + (v p_1 - u p_2)\cos \theta).
\end{equation}

Relative to the basis $(J,D,P_1,P_2,H)$ for $\g$,
the matrix of $\widehat \Ad _{(k,h)}$ is given by
\begin{equation}
  \widehat\Ad_{(k,h)} =
  \begin{pmatrix}
    1 & 0 & 0 & 0 & 0\\
    0 & 1 & 0 & 0 & 0\\
    p_2 & -w p_1 & e^{w\sigma}\cos\theta & - e^{w\sigma} \sin\theta & 0 \\
    -p_1 & -w p_2 & e^{w\sigma} \sin\theta & e^{w\sigma} \cos\theta & 0 \\
    -\tfrac12(p_1^2+p_2^2) & w (p_1p_2 - 2h) & e^{w\sigma}(p_1 \sin\theta - p_2\cos\theta) & e^{w\sigma}(p_2 \sin\theta + p_1 \cos\theta) & e^{2w\sigma}
  \end{pmatrix}
\end{equation}
and hence relative to the canonical dual basis $(\lambda, \delta, \pi^1, \pi^2, \eta)$ for $\g^*$, the matrix of the
coadjoint action $\widehat\Ad^*_{(k,h)}$ is given by the inverse
transpose of the above matrix:
\begin{equation}
  \widehat\Ad^*_{(k,h)} =
  \begin{pmatrix}
    1 & 0 & e^{-w \sigma}(p_1 \sin\theta - p_2 \cos\theta) & e^{-w \sigma} (p_1 \cos\theta + p_2 \sin\theta) & -\tfrac12 e^{-2 w \sigma} (p_1^2 + p_2^2)\\
    0 & 1 & e^{-w \sigma} w (p_1 \cos\theta + p_2 \sin\theta) & e^{-w \sigma} w (p_2 \cos\theta - p_1 \sin\theta) & e^{-2 w \sigma} w (2h - p_1 p_2) \\
    0 & 0 & e^{-w \sigma} \cos\theta & -e^{-w \sigma} \sin\theta & e^{-2 w \sigma} p_2\\
    0 & 0 & e^{-w \sigma} \sin\theta & e^{-w \sigma} \cos\theta & -e^{-2 w \sigma} p_1\\
    0 & 0 & 0 & 0 &  e^{-2 w \sigma}
  \end{pmatrix},
\end{equation}
so that the coadjoint action by $(k,h)$ given in
equation~\eqref{eq:hX-defs} on the dual basis is explicitly:
\begin{equation}
  \begin{split}
    \lambda &\mapsto \lambda\\
    \delta &\mapsto \delta\\
    \pi^1 &\mapsto e^{-w\sigma}\left( \pi^1 \cos\theta + \pi^2 \sin\theta + (p_1
      \sin\theta - p_2 \cos\theta) \lambda + w (p_1 \cos\theta + p_2 \sin\theta) \delta\right)\\
    \pi^2 &\mapsto e^{-w\sigma}\left( -\pi^1 \sin\theta + \pi^2 \cos\theta + (p_1
      \cos\theta + p_2 \sin\theta) \lambda + w (p_2 \cos\theta - p_1 \sin\theta)
      \delta\right)\\
    \eta &\mapsto e^{-2w\sigma} \left(\eta -p_1 \pi^2 + p_2 \pi^1 - \tfrac12
    (p_1^2+p_2^2)\lambda + w (2h - p_1p_2) \delta\right).
  \end{split}
\end{equation}
Let us introduce coordinates $x_\lambda, x_\delta, x_{\pi^1},
x_{\pi^2}, x_\eta$ for $\g^*$.  Then the infinitesimal
generators of the coadjoint representation are the following vector
fields on $\g^*$:
\begin{equation}
  \begin{split}
    \xi_J &= - x_{\pi^2} \frac{\d}{\d x_{\pi^1}} + x_{\pi^1} \frac{\d}{\d x_{\pi^2}}\\
    \xi_D &= - w x_{\pi^1} \frac{\d}{\d x_{\pi^1}} - w x_{\pi^2} \frac{\d}{\d x_{\pi^2}} - 2  w x_{\eta} \frac{\d}{\d x_\eta}\\
    \xi_{P_1} &= x_{\pi^2} \frac{\d}{\d x_\lambda}   + w x_{\pi^1} \frac{\d}{\d x_\delta} - x_\eta \frac{\d}{\d x_{\pi^2}}\\
    \xi_{P_2} &= - x_{\pi^1} \frac{\d}{\d x_\lambda} + w x_{\pi^2} \frac{\d}{\d x_\delta} + x_\eta \frac{\d}{\d x_{\pi^1}}\\
    \xi_H &= 2 w x_\eta \frac{\d}{\d x_\delta}.
  \end{split}
\end{equation}
One can check that $\xi : \g \to \eX(\g^*)$ given by
$X \mapsto \xi_X$ is a Lie algebra antihomomorphism, as expected.

\subsubsection{Coadjoint orbits for $w=0$}
\label{sec:coadjoint-orbits-1}

Let us set $w=0$.  The coadjoint action on the coordinates $(x_\lambda, x_\delta, x_{\pi^1},x_{\pi^2}, x_\eta)$ is then given by
\begin{equation}
  \begin{split}
    x_\lambda & \mapsto x_\lambda + (p_1 \sin\theta - p_2 \cos\theta) x_{\pi^1} + (p_1 \cos\theta + p_2 \sin\theta ) x_{\pi^2} - \tfrac12 (p_1^2+p_2^2) x_\eta\\
    x_\delta & \mapsto x_\delta \\
    x_{\pi^1} &\mapsto \cos\theta x_{\pi^1} - \sin\theta x_{\pi^2} + p_2 x_\eta\\
    x_{\pi^2} &\mapsto \cos\theta x_{\pi^2} + \sin\theta x_{\pi^1} - p_1 x_\eta\\
    x_\eta &\mapsto x_\eta.
  \end{split}
\end{equation}
We must distinguish between two cases, depending on whether or not
$x_\eta$, which is inert, vanishes.
\begin{enumerate}
\item If $x_\eta = 0$, the orbit is either
  \begin{enumerate}
  \item a two-dimensional cylinder in the affine 3-plane defined by giving a
    constant value to $x_\delta$ and setting $x_\eta$ to
    zero, if at least one of $x_{\pi^1}$ and $x_{\pi^2}$ is nonzero; or
  \item a point with coordinates $(x_\lambda,x_\delta,0,0,0)$ if
    $x_{\pi^1} = x_{\pi^2} = 0$.
  \end{enumerate}
\item If $x_\eta \neq 0$, we have a two-dimensional surface in the
  affine 3-plane with constant $(x_\delta, x_\eta\neq 0)$,
  which is obtained as the graph $x_\lambda = f(x_{\pi^1},x_{\pi^2})$
  of a function in the ($x_{\pi^1},x_{\pi^2}$)-plane and hence
  diffeomorphic to $\RR^2$.
\end{enumerate}

\subsubsection{Coadjoint orbits for $w=1$}
\label{sec:coadjoint-orbits-2}

Now let's consider $w=1$.  The coadjoint action on the coordinates is
now
\begin{equation}
  \begin{split}
    x_\lambda &\mapsto x_\lambda + (p_1 \sin\theta - p_2 \cos\theta) e^{-\sigma} x_{\pi^1} + (p_1 \cos\theta + p_2 \sin\theta) e^{-\sigma} x_{\pi^2} - \tfrac12 (p_1^2 + p_2^2) e^{-2\sigma} x_\eta\\
    x_\delta &\mapsto x_\delta + (p_1 \cos\theta + p_2 \sin\theta ) e^{-\sigma} x_{\pi^1} + (p_2 \cos\theta - p_1 \sin\theta) e^{-\sigma} x_{\pi^2} +  (2h - p_1 p_2) e^{-2\sigma} x_\eta\\
    x_{\pi^1} &\mapsto e^{-\sigma} \cos\theta x_{\pi^1} - e^{-\sigma} \sin\theta x_{\pi^2} + p_2 e^{-2\sigma} x_\eta \\
    x_{\pi^1} &\mapsto e^{-\sigma} \sin\theta x_{\pi^1} + e^{-\sigma} \cos\theta x_{\pi^2} - p_1 e^{-2\sigma} x_\eta \\
    x_\eta & \mapsto e^{-2\sigma} x_\eta.
  \end{split}
\end{equation}
We must again distinguish between two cases:
\begin{enumerate}
  \item if $x_\eta = 0$, the orbits are either
  \begin{enumerate}
  \item points with
    coordinates $(x_\lambda,x_\delta, 0,0,0)$, if
    $(x_{\pi^1},x_{\pi^2})= (0,0)$; or
  \item if $(x_{\pi^1},x_{\pi^2}) \neq (0,0)$, the orbit is
    four-dimensional and consists of the $4$-plane $x_\eta =
    0$ with the $2$-plane with additional equations $x_{\pi^1} =
    x_{\pi^2} = 0$ removed.  So diffeomorphic to $\RR^4 \setminus
    \RR^2 \cong \RR^2 \times (\RR^2 \setminus \{(0,0)\})$.
  \end{enumerate}
  \item if $x_\eta \neq 0$, the orbit is a
    four-dimensional hypersurface given by the graph of a function of
    the four coordinates $(x_\lambda,x_\delta,x_{\pi^1}, x_{\pi^2})$.
\end{enumerate}

\section{Conclusion}
\label{sec:conclusion}

This work provides the first systematic classification of Lifshitz
algebras, spacetimes and particles. We also provide a full
classification of aristotelian spacetimes with scalar charge, in
particular ones with exotic spacetime symmetries. The Lifshitz
algebras are summarised in Table~\ref{tab:lifshitz} and the respective
spacetimes fall into three classes: $(d+2)$-dimensional Lifshitz
spacetimes (Table~\ref{tab:d+2-lifshitz-spaces-sum}) where the
dilatations provide an additional holographic direction and the
$(d+1)$-dimensional Lifshitz--Weyl spacetimes
(Table~\ref{tab:d+1-lifshitz-spaces-eff}) and $(d+1)$-dimensional
aristotelian spacetimes with a scalar charge
(Table~\ref{tab:d+1-exot}). It is interesting to note that $\a_{2}$,
$\a_{3}^{z\neq 0}$ and $\a_{5}^{\pm}$ give rise to each of the above
discussed spacetimes and can therefore be interpreted from various
different angles. We refer to Section~\ref{sec:summary-results} for a
summary of our results and end with a few remarks.
\begin{description}
\item[Beyond the standard Lifshitz symmetries] Our classification was
  motivated by the standard Lifshitz symmetries and it is reassuring
  that we indeed recover them as spacetimes $\zLW3_{z}$ and
  $\zLW2_{z}$. Beyond this case let us highlight the pairs based
  on $\a_{2}$ ($\zL2$ and $\zLW1$) and its curved
  generalisation $\a_{5}^{\pm}$ ($\zL5$ and $\zLW3_{\pm}$)
  both of which exist in generic dimension. It might well be that
  these spaces have played a rôle in the literature, if so we are
  unfortunately unaware of it.
\item[Exotic spacetime symmetries] Exotic spacetime symmetries,
  similar to the ones discussed in Section~\ref{sec:summary-results}
  have recently played a rôle in relation to fractons (see,
  e.g.,~\cite{Nandkishore:2018sel,Pretko:2020cko,Grosvenor:2021hkn}
  for reviews and
  \cite{Nicolis:2008in,Griffin:2013dfa,Hinterbichler:2014cwa,Griffin:2014bta}
  for earlier related work on polynomial shift symmetries). In the
  case of fractons the underlying geometry is also
  aristotelian~\cite{Bidussi:2021nmp,Jain:2021ibh} and there is also
  an action of the spacetime symmetry on the charges. There has been a
  systematic study of exotic symmetries of this
  type~\cite{Gromov:2018nbv}, however the symmetries we discuss here
  are a generalisation of the multipole
  algebras~\cite{Gromov:2018nbv}. We also allow for nontrivial
  commutation relations between the temporal, rather than just the
  spatial, translations and the underlying aristotelian geometry is
  not necessarily restricted to be flat. It is this generalisation
  that leads to the novel classification of exotic symmetries of
  Table~\ref{tab:d+1-exot}. Another consequence of our classification is
  that there are no aristotelian spacetimes with one scalar charge and
  nonzero $[Q,\P]$ commutator, like for the more conventional
  multipole algebras. With regard to~\cite{Gromov:2018nbv} we have
  kept the full rotational symmetry, which is an assumption that might
  be dropped and would lead to a more general classification.
\item[Exotic aristotelian theories with scalar charge] By showing that
  these symmetries exist and can be consistently realised we have
  provided the first nontrivial step for the construction of exotic
  aristotelian theories with scalar charge. However, to further
  clarify their physical relevance and the relation to fracton-like
  theories it would be interesting to construct field theories that
  realise these symmetries (e.g., following~\cite{Hirono:2022dci}).
\item[Generalisation] In this work the spacetimes have no boost
  symmetry so a natural generalisation is to add an additional vector
  to our classification. This leads for example to Bargmann spacetimes
  and will be discussed in a future
  work~\cite{Figueroa-OFarrill:20222xxx}.

  The case of aristotelian geometries with one vector charge and no
  additional scalar has already been worked out
  in~\cite{Figueroa-OFarrill:2018ilb}, the interpretation as exotic
  charges had however not been appreciated at that point.
\item[Unitary irreducible representations] Upon quantisation one
  expects that coadjoint orbits lead to unitary irreducible
  representations. In this sense our classification of the coadjoint
  orbits of the Lifshitz groups lays the foundation for such an
  endeavour and it might be interesting to understand if and in which
  sense these representations can be understood as quantum Lifshitz
  particles.
\end{description}

\section*{Acknowledgments}
\label{sec:acknowledgments}

We are grateful to Leo Bidussi, Joaquim Gomis, Kevin Grosvenor, Jelle
Hartong, Emil Have, Yuji Hirono, Kristan Jensen, Jørgen Musaeus,
Blagoje Oblak and Alfredo Pérez for useful discussions. We also thank
Anton Galajinsky for bringing \cite{Galajinsky:2022bwq} and Mohammad
Reza Mohammadi Mozaffar for bringing~\cite{MohammadiMozaffar:2017nri}
and other interesting works to our attention.

During the start of this project the research of SP was supported by
the ERC Advanced Grant ``High-Spin-Grav'' and by FNRS-Belgium
(Convention FRFC PDR T.1025.14 and Convention IISN 4.4503.15). SP was
supported by the Leverhulme Trust Research Project Grant
(RPG-2019-218) ``What is Non-Relativistic Quantum Gravity and is it
Holographic?''.

SP acknowledges support of the Erwin Schrödinger Institute (ESI) in
Vienna where part of this work was conducted during the thematic
programme ``Geometry for Higher Spin Gravity: Conformal Structures,
PDEs, and Q-manifolds''.

\appendix

\section{Coadjoint orbits of the two-dimensional nonabelian Lie group}
\label{sec:coadjoint-orbits-G2}

The connected (and simply-connected) Lie group $\GG$ whose Lie algebra
$\g = \left<D,H\right>$ with bracket $[D,H]=H$ is isomorphic to the
matrix group
\begin{equation}
  \GG = \left\{
    \begin{pmatrix}
      a & b \\ 0 & 1
    \end{pmatrix}
\middle | a,b \in \RR,~ a> 0\right\},
\end{equation}
with Lie algebra
\begin{equation}
  \g = \left\{
    \begin{pmatrix}
      x & y \\ 0 & 0
    \end{pmatrix}
    \middle | x,y \in \RR\right\}.
\end{equation}
We introduce the (non-invariant) inner product $\left<-,-\right>$ on
$\g$ by
\begin{equation}
  \left<X,Y\right> = \tr X^T Y~.
\end{equation}
This inner product defines a musical (vector space) isomorphism $\flat
: \g \to \g^*$.   If $X^\flat \in \g^*$ and $g \in \GG$ we define the
coadjoint action by
\begin{equation}
  (\Ad^*_g X^\flat)(Y) = X^\flat (\Ad_{g^{-1}} Y) = \tr X^T g^{-1} Y g
  = \tr g X^T g^{-1} Y.
\end{equation}
Explicitly, if
\begin{equation}
  X^T =
  \begin{pmatrix}
    \alpha & 0 \\ \beta & 0
  \end{pmatrix}
\end{equation}
then under the coadjoint action of $g =\begin{pmatrix}  a & b \\ 0 &
  1\end{pmatrix}$,
\begin{equation}
  \alpha \mapsto \alpha + b a^{-1} \beta \qquad\text{and}\qquad \beta
  \mapsto \beta a^{-1}.
\end{equation}
Therefore the orbit of $(\alpha,\beta) = (\alpha, 0)$ is a point,
whereas the orbit of $(\alpha, \beta > 0)$ is the half-plane $\beta >
0$ and that of $(\alpha, \beta < 0)$ is the half-plane $\beta < 0$.

\section{Coadjoint orbits of the Heisenberg group}
\label{sec:coadj-orbits-heis}

The Heisenberg group $\NN$ is the three-dimensional Lie group of
strictly upper triangular unipotent $3 \times 3$ matrices
\begin{equation}
  \NN = \left\{
    \begin{pmatrix}
      1 & a & c \\ 0 & 1 & b \\ 0 & 0 & 1
    \end{pmatrix}
\middle | a,b,c \in \RR\right\},
\end{equation}
with Lie algebra
\begin{equation}
  \n = \left\{
    \begin{pmatrix}
      0 & x & z \\ 0 & 0 & y \\ 0 & 0 & 0
    \end{pmatrix}
\middle | x,y,z \in \RR\right\}.
\end{equation}
It is straightforward to work out the adjoint action of $\NN$ on $\n$:
if $g =\begin{pmatrix} 1 & a & c \\ 0 & 1 & b \\ 0 & 0 &
  1 \end{pmatrix}$,
\begin{equation}
  \Ad_g
  \begin{pmatrix}
      0 & x & z \\ 0 & 0 & y \\ 0 & 0 & 0
    \end{pmatrix}
    =
    \begin{pmatrix}
      0 & x & z + a y - b x \\ 0 & 0 & y \\ 0 & 0 & 0
    \end{pmatrix}.
\end{equation}
We can identify $\n^*$ with $\n$ as vector spaces under the musical
isomorphism of the non-invariant inner product $\left<-,-\right>$ on
$\h$ given by $\left<X,Y\right> = \tr X^T Y$.  Then under the same
element $g$ as before, the coadjoint action is given by
\begin{equation}
  \Ad^*_g \begin{pmatrix}
      0 & \alpha & \gamma \\ 0 & 0 & \beta \\ 0 & 0 & 0
    \end{pmatrix}
    =
    \begin{pmatrix}
      0 & \alpha + b \gamma  & \gamma \\ 0 & 0 & \beta - a \gamma \\ 0 & 0 & 0
    \end{pmatrix}.
\end{equation}
Therefore there are two kinds of coadjoint orbits:
\begin{itemize}
\item the orbits of $(\alpha,\beta,\gamma) = (\alpha, \beta, 0)$, which
  are points; and
\item the orbits of $(\alpha,\beta, \gamma \neq 0)$, which are the
  affine (hyper)planes $\left\{(\alpha,\beta,\gamma)\middle | \alpha,\beta
    \in \RR\right\}$.
\end{itemize}

\section{Coadjoint orbits of $\SO(n)$}
\label{sec:coadjoint-orbits-son}

In this appendix we describe in some more detail the adjoint orbits of
$\SO(n)$.

Let $\Lambda \in \so(n)$. Since $\Lambda$ is skew-symmetric,
$i\Lambda$ is a hermitian endomorphism of $\CC^n$ and thus has real
eigenvalues.  This means $\Lambda$ has purely imaginary eigenvalues
and since $\Lambda$ is real, they are either zero or else come in
complex conjugate pairs $(i\lambda,-i\lambda)$. Let $\lambda \neq 0$
and $E_{i\lambda} \subset \CC^n$ be the eigenspace of $\Lambda$ with
eigenvalue $i\lambda$. If $v \in E_{i\lambda}$, its complex conjugate
$\bar v \in E_{-i\lambda}$. Write $v = x + i y$, where
$x,y \in \RR^n$. Then $\Lambda x = - \lambda y$ and
$\Lambda y = \lambda x$. If $i\lambda$ has multiplicity $m$, then we
can choose a unitary basis $v_1,\dots,v_m$ for $E_{i\lambda}$ and
writing $v_j = x_i + i y_j$, we have that on the span of
${x_1,\dots,x_m,y_1,\dots,y_m}$, the matrix representing $\Lambda$
takes the form of $m$ blocks of $2\times 2$ matrices of the form
\begin{equation}
  \lambda J := \begin{pmatrix}    0 & -\lambda \\ \lambda & 0 \end{pmatrix}.
\end{equation}
In this way, we find that there is a basis for $\RR^n$ relative to
which $\Lambda$ has matrix
\begin{equation}
  \begin{pmatrix}
    0 & & & \\
    & \lambda_1 J_{m_1} & & \\
    & & \ddots & \\
    & & & \lambda_k J_{m_k}
  \end{pmatrix}
\end{equation}
where $m_0 = \dim\ker \Lambda$ is the size of the first block, $m_j$
is the multiplicity of $i\lambda_j$ and $J_m$ is a $2m \times 2m$
matrix consisting of $m$ blocks of $2\times 2$ matrices which are
either
\begin{equation}
  J = \begin{pmatrix}
    0 & -1 \\ 1 & 0 
  \end{pmatrix} \qquad\text{or}\qquad
  - J = \begin{pmatrix}
    0 & 1 \\ -1 & 0
  \end{pmatrix}.
\end{equation}
It is clear that the transformation taking $\Lambda$ to the above
matrix is orthogonal, but perhaps not special orthogonal.  Indeed, the
matrices $\pm J$ are related by an orientation reversing
transformation of the plane, and hence since we are only allowed to
conjugate by $\SO(n)$, we cannot simply make all blocks be $J$,
say. This problem only arises for $n$ even, since if $n$ is odd, then
$\Lambda$ has kernel and we can indeed use the freedom to multiply one
of the basis elements in the kernel by $-1$ to ensure that the
orthogonal matrix which conjugates $\Lambda$ to a normal form where
all $2\times 2$ blocks are of the form $J$, has determinant $1$.

As an example, let us list the possible coadjoint orbits of $SO(4)$:
\begin{itemize}
\item If $\Lambda = 0$, we get a point-like orbit with stabiliser
  $\SO(4)$.
\item If $\dim\ker\Lambda = 2$, then we can bring $\Lambda$ to the
  form
  \begin{equation}
    \begin{pmatrix}
      0 & 0 & 0 & 0\\
      0 & 0 & 0 & 0\\
      0 & 0 & 0 & -\lambda\\
      0 & 0 & \lambda & 0
    \end{pmatrix}
  \end{equation}
  for some real $\lambda \neq 0$.  We can arrange for $\lambda > 0$
  without loss of generality, since, e.g.,
  \begin{equation}
    \begin{pmatrix}
      1 & 0 & 0 & 0 \\
      0 & -1 & 0 & 0 \\
      0 & 0 & 0 & 1\\
      0 & 0 & 1 & 0
    \end{pmatrix} \in \SO(4)
  \end{equation}
  and conjugating by that matrix changes $\lambda$ to $-\lambda$.  The
  stabiliser in this case is $\operatorname{S(O(2) \times U(1))}$.
\item If $\dim \ker \Lambda = 0$, we have two possibilities:
  \begin{itemize}
  \item  if the eigenvalue $\lambda$ has multiplicity $2$, then we
    have two kinds of orbits corresponding to
    \begin{equation}
      \begin{pmatrix}
        \lambda J & 0 \\
        0 & \lambda J
      \end{pmatrix}
      \qquad\text{or}\qquad
      \begin{pmatrix}
        \lambda J & 0 \\
        0 & -\lambda J
      \end{pmatrix}.
    \end{equation}
    Those two matrices are conjugate in $\operatorname{O(4)}$, but not
    in $\SO(4)$.  The stabiliser of these coadjoint orbits is $\U(2)
    \subset \SO(4)$.
  \item if the two eigenvalues $\lambda_1,\lambda_2$ are different,
    then we can assume without loss of generality that $\lambda _1
    \leq \lambda_2$, and we can bring $\Lambda$ to one of the
    following two matrices
    \begin{equation}
      \begin{pmatrix}
        \lambda_1 J & 0 \\ 0 & \pm \lambda_2 J
      \end{pmatrix}.
    \end{equation}
    The stabiliser in this case is $\operatorname{S(U(1)\times U(1))}$.
  \end{itemize}
\end{itemize}

Alternatively, we can argue as follows.  Given any $\Lambda \in
\so(n)$, it is conjugate under $\SO(n)$ to a matrix in a Cartan
subalgebra, say:
\begin{equation}
  \begin{pmatrix}
    0 & \lambda_1 & & & & \\
    -\lambda_1 & 0 & & & & \\
    & & 0 & \lambda_2 & & \\
    & & -\lambda_2 & 0 & & \\
    & & & & \ddots & \\
    & & & & & 0
  \end{pmatrix},
\end{equation}
where the $\lambda_i$ are not necessarily distinct. This still leaves
the action of the Weyl group. If $n= 2\ell + 1$ is odd, then the Weyl
group is $\{\pm 1\}^\ell \rtimes S_\ell$: the symmetric group $S_\ell$
acts by permuting the $\lambda_i$ and $\{\pm 1\}^\ell$ acts by
changing the signs of the $\lambda_i$. We can therefore arrange them
so that
$0 \leq \lambda_1 \leq \lambda_2 \leq \cdots \leq \lambda_\ell$. If
$n = 2 \ell$ is even, then the Weyl group is the index-2 subgroup of
$\{\pm 1\}^\ell \rtimes S_\ell$ consisting of an \emph{even} number of
sign changes, equivalently, it is the subgroup
$K_\ell \rtimes S_\ell$, where $K_\ell$ is the kernel of the group
homomorphism $\{\pm 1\}^\ell \to \{\pm 1\}$ sending
$(\sigma_1,\dots,\sigma_\ell) \mapsto \sigma_1\dots \sigma_\ell$.

In the case of $\SO(4)$ we can bring any $\Lambda \in \so(4)$ to
\begin{equation}
  \begin{pmatrix}
    0 & \lambda_1 & & \\
    -\lambda_1 & 0 & & \\
    & & 0 & \lambda_2 \\
    & & -\lambda_2 & 0
  \end{pmatrix},
\end{equation}
and the Weyl group is the Klein \emph{Vierergruppe} acting on
$(\lambda_1, \lambda_2)$ by permuting them
$(\lambda_1, \lambda_2) \mapsto (\lambda_2, \lambda_1)$ or changing
\emph{both} their signs
$(\lambda_1, \lambda_2) \mapsto (-\lambda_1, -\lambda_2)$. The moduli
space of coadjoint orbits is then the wedge
$-\lambda_2 \leq \lambda_1 \leq \lambda_2$, which is illustrated in
Figure~\ref{fig:coad-orbs-so4}. The generic stabiliser is
$\operatorname{S(U(1)\times U(1))}$. If $\lambda_1 = 0$, but
$\lambda_2 \neq 0$, then the stabiliser is
$\operatorname{S(O(2) \times U(1))}$. If
$\lambda_1 = \pm \lambda_2 \neq 0$, then the stabiliser is enhanced to
$\U(2)$, and if $\lambda_1 = \lambda_2 = 0$, it is all of $\SO(4)$.

\begin{figure}[h!]
  \centering
  \begin{tikzpicture}[>=latex, x=1cm,y=1cm]
    %
    %
    \fill[blue!20!white] (-3,3) -- (0,0) -- (3,3) -- cycle;
    %
    %
    \draw[black!70!white] (-4,0) -- (4,0) node [above] {$\lambda_1$};
    \draw[black!70!white] (0,-1) -- (0,4) node [right] {$\lambda_2$};
    %
    %
    \draw[blue!20!black, thick] (0,0) -- (3,3) node [pos=0.5, below, sloped] {\tiny $\lambda_1 = \lambda_2$};
    \draw[blue!20!black, thick] (-3,3) -- (0,0) node [pos=0.5, below, sloped] {\tiny $\lambda_1 = -\lambda_2$};
    %
    %
    \filldraw [color=red!70!black,fill=red!50!white] (0,0) circle (2pt);
  \end{tikzpicture}
  \caption{Moduli space of coadjoint orbits of $\SO(4)$}
  \label{fig:coad-orbs-so4}
\end{figure}
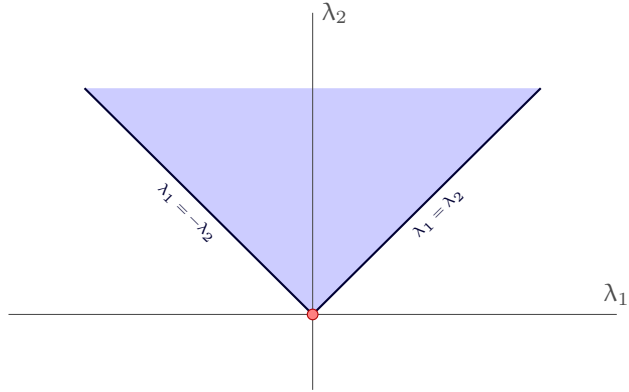

\section{Coadjoint orbits of the Lorentz group}
\label{sec:coadjoint-orbits-lorentz}

In this appendix we describe the coadjoint orbits of the (proper,
orthochronous) Lorentz group; that is, the identity component
$\SO(n,1)_0$ of the ($n+1$)-dimensional Lorentz group, for $n\geq
2$. Since for those values of $n$, $\so(n,1)$ is semisimple, the
musical isomorphisms associated to the Killing form, which are
$\so(n,1)$-equivariant, set up an equivalence between the adjoint and
coadjoint representations and hence between the adjoint and coadjoint
orbits under the identity component $\SO(n,1)_0$.

As explained, for example, in a different context in
\cite[Section~2]{Figueroa-OFarrill:2004lpm}, every $\Lambda \in
\so(p,q)$ can be brought to a normal form under the adjoint action of
$\SO(p,q)$.  In that normal form, $\Lambda$ is block diagonal where
the relevant elementary blocks for $(p,q) = (n,1)$ are given in
Table~\ref{tab:blocks}, together with their signature $(p,q)$.  To
obtain the normal forms for $\Lambda \in \so(n,1)$, we need to build
all block diagonal matrices of signature $(n,1)$ using the elementary
blocks and then check whether there is a further identification
of the parameters under the adjoint action.

\begin{table}[h!]
  \centering
  \caption{Elementary blocks for lorentzian signature}
  \label{tab:blocks}
  \setlength{\extrarowheight}{3pt}
  \begin{tabular}{>{$}l<{$}|>{$}l<{$}}
    \multicolumn{1}{c|}{Signature $(p,q)$} & \multicolumn{1}{c}{Block in $\so(p,q)$}\\
    \toprule
    (1,0) & [0]\\
    (0,1) & [0]\\
    (2,0) & \begin{bmatrix} 0 & \varphi \\ -\varphi & 0\end{bmatrix}\\
    (1,1) & \begin{bmatrix} 0 & \beta \\ \beta & 0 \end{bmatrix}\\
    (2,1) & \begin{bmatrix} 0 & 1 & 0 \\ 1 & 0 & \pm 1 \\ 0 & \mp 1 &  0\end{bmatrix}\\
    \bottomrule
  \end{tabular}
  \caption*{The parameters $\varphi,\beta$ are nonzero real numbers,
    but depending on $n$, we may restrict their sign.}
\end{table}

For example, if $n=2$, then we can build a block-diagonal
matrix of signature $(2,1)$ from elementary blocks as follows:
\begin{itemize}
\item $(0,1) \oplus (1,0) \oplus (1,0)$
\item $(0,1) \oplus (2,0)$
\item $(1,1) \oplus (1,0)$
\item $(2,1)$
\end{itemize}
This results in the following normal forms for matrices in $\so(2,1)$:
\begin{equation}
  \begin{pmatrix}
    0 & 0 & 0 \\ 0 & 0 & 0 \\ 0 & 0 & 0
  \end{pmatrix}\qquad
  \begin{pmatrix}
    0 & 1 & 0 \\ 1 & 0 & \pm 1 \\ 0 & \pm 1 & 0
  \end{pmatrix}\qquad
  \begin{pmatrix}
    0 & 0 & 0 \\ 0 & 0 & \varphi \\ 0 & -\varphi & 0
  \end{pmatrix}\qquad
  \begin{pmatrix}
    0 & \beta & 0 \\ \beta & 0 & 0\\ 0 & 0 & 0
  \end{pmatrix}.
\end{equation}
The two irreducible blocks cannot be related under the identity
component $SO(n,1)_0$, and neither can the sign of $\varphi$ be
changed.  However, there is a proper orthochronous Lorentz
transformation which changes the sign of $\beta$ in the last block, so
that we can actually take $\beta > 0$.  Because the Killing form for
$\so(2,1)$ has signature $(2,1)$, the coadjoint orbits of $\SO(2,1)_0$
coincide with the orbits in a three-dimensional lorentzian vector
$(V,\eta)$ space under the proper orthochronous Lorentz
transformations: namely,
\begin{itemize}
\item the origin, corresponding to the zero matrix;
\item the future and past deleted lightcones $\eta(v,v)=0$, corresponding to the
  matrices $\begin{pmatrix} 0 & 1 & 0 \\ 1 & 0 & \pm 1 \\ 0 & \pm 1 &
    0 \end{pmatrix}$;
\item the upper and lower sheets of the two-sheeted hyperboloids
  $\eta(v,v)=-\varphi^2$, corresponding to the matrices $\begin{pmatrix}
    0 & 0 & 0 \\ 0 & 0 & \varphi \\ 0 & -\varphi & 0
  \end{pmatrix}$ for $\pm \varphi > 0$; and
\item the one-sheeted hyperboloid $\eta(v,v)=\beta^2$, corresponding
  to the matrices $\begin{pmatrix} 0 & \beta & 0 \\ \beta & 0 & 0\\ 0
    & 0 & 0 \end{pmatrix}$, for $\beta > 0$.
\end{itemize}
This of course is completely elementary and is the lorentzian analogue
of the decomposition of three-dimensional euclidean space under the
group of rotations into the origin and the spheres of radius $r$ for
$r>0$.

To illustrate further, let us now consider $n=3$ and discuss the
coadjoint orbits of $\SO(3,1)_0$.  Using the elementary blocks in
Table~\ref{tab:blocks}, we can build a block diagonal matrix of
signature $(3,1)$ as follows:
\begin{itemize}
\item $(0,1) \oplus (1,0) \oplus (1,0) \oplus (1,0)$
\item $(2,1) \oplus (1,0)$
\item $(1,1) \oplus (2,0)$
\item $(1,1) \oplus (1,0) \oplus (1,0)$
\item $(0,1) \oplus (1,0) \oplus (2,0)$
\end{itemize}
This results in the following normal forms for matrices in $\so(3,1)$:
\begin{equation}
  \begin{pmatrix}
    0 & 1 & 0 & 0 \\
    1 & 0 & \pm 1 & 0\\
    0 & \mp 1 & 0 & 0\\
    0 & 0 & 0 & 0
  \end{pmatrix}\qquad
  \begin{pmatrix}
    0 & \beta & 0 & 0\\
    \beta & 0 & 0 & 0\\
    0 & 0 & 0 & \varphi\\
    0 & 0 & -\varphi & 0
  \end{pmatrix}\qquad
  \begin{pmatrix}
    0 & \beta & 0 & 0\\
    \beta & 0 & 0 & 0\\
    0 & 0 & 0 & 0\\
    0 & 0 & 0 & 0
  \end{pmatrix}\qquad
  \begin{pmatrix}
    0 & 0 & 0 & 0\\
    0 & 0 & 0 & 0\\
    0 & 0 & 0 & \varphi\\
    0 & 0 & -\varphi & 0
  \end{pmatrix}
\end{equation}
in addition to the zero matrix.  Now conjugation with $\SO(3,1)_0$
relates the versions of the first matrix with the different signs, so
we need only consider one of them.  Similarly, in the last two normal
forms, we need only keep $\beta >0$ and $\varphi>0$.  In the second
normal form, conjugation changes the signs of $\beta$ and $\varphi$
simultaneously, hence we can take $\beta > 0$ and $\varphi$ nonzero
but otherwise unconstrained.

\providecommand{\href}[2]{#2}\begingroup\raggedright\endgroup


\end{document}